\documentclass[runningheads]{llncs}

\usepackage{cite}
\usepackage{amsmath,amssymb,amsfonts}
\usepackage{graphicx}
\usepackage{fullpage}
\usepackage{textcomp}
\usepackage{xcolor}
\def\BibTeX{{\rm B\kern-.05em{\sc i\kern-.025em b}\kern-.08em
		T\kern-.1667em\lower.7ex\hbox{E}\kern-.125emX}}

\usepackage[T1]{fontenc}
\def\doi#1{\href{https://doi.org/\detokenize{#1}}{\url{https://doi.org/\detokenize{#1}}}}
\usepackage{graphicx}

\newif\ifFull
\Fullfalse

\usepackage{graphicx}
\usepackage{balance}  %

\usepackage{caption}
\usepackage{subcaption}

\usepackage[utf8]{inputenc}
\usepackage{algorithm}

\usepackage{numprint}
\usepackage[noend]{algpseudocode}

\usepackage{hyperref}
\usepackage{xspace}
\usepackage{wrapfig}
\usepackage{graphics}
\usepackage{booktabs}
\usepackage{microtype}%
\usepackage[left,pagewise] {lineno}
\definecolor {infocolor} {rgb} {0.6,0.6,0.6}

\usepackage{multicol}%

\usepackage{multirow}

\usepackage{tikz}

\usepackage{pgfplots}

\usepackage{framed,enumitem}

\newcommand{\Is}       {:=}
\newcommand{\set}[1]{\left\{ #1\right\}}

\newcommand{\gilt}{:}
\newcommand{\sodass}{\,:\,}
\newcommand{\setGilt}[2]{\left\{ #1\sodass #2\right\}}

\def\MdR{\ensuremath{\mathbb{R}}}

\newcommand{\Id}[1]{\texttt{\detokenize{#1}}}

\newcommand{\ie}{i.\,e.,\xspace}
\newcommand{\eg}{e.\,g.,\xspace}

\usepackage{color}

\def\comment#1{}

\def\withcomments{
	\newcounter{mycommentcounter}
	\def\comment##1{\refstepcounter{mycommentcounter}%
		\ifhmode%
		\unskip%
		{\dimen1=\baselineskip \divide\dimen1 by 2 %
			\raise\dimen1\llap{\tiny\bfseries \textcolor{red}{-\themycommentcounter-}}}\fi%
		\marginpar[{\renewcommand{\baselinestretch}{0.8}%
			\hspace*{3em}\begin{minipage}{5em}\footnotesize [\themycommentcounter]: \raggedright ##1\end{minipage}}]{\renewcommand{\baselinestretch}{0.8}%
			\begin{minipage}{5em}\footnotesize [\themycommentcounter]: \raggedright ##1\end{minipage}}}
}

\definecolor{darkgreen}{RGB}{0,200,100}
\definecolor{orange}{RGB}{255,80,0}

\newcommand{\Xcomment}[1]{}

\usepackage{listings}
\usepackage[square,numbers]{natbib}
\begin{document}
\title{Recursive Multi-Section on the Fly:\\Shared-Memory Streaming Algorithms for Hierarchical~Graph~Partitioning~and~Process~Mapping\thanks{Partially supported by DFG grant SCHU 2567/1-2.}}

\author{Marcelo Fonseca Faraj\inst{1} \and
Christian Schulz\inst{2}}
\authorrunning{Faraj and Schulz}
\institute{Heidelberg University, Heidelberg, Germany, 
\email{marcelofaraj@informatik.uni-heidelberg.de} \and Heidelberg University, Heidelberg, Germany,
\email{christian.schulz@informatik.uni-heidelberg.de}
 }
\maketitle              %
\begin{abstract}
Partitioning a graph into balanced blocks such that few edges run between blocks is a key problem for large-scale distributed processing. A current trend for partitioning huge graphs are streaming algorithms, which use low computational resources. In this work, we present a shared-memory streaming multi-recursive partitioning scheme that performs recursive multi-sections on the fly without knowing the overall input graph. Our approach has a considerably lower running time complexity in comparison with state-of-the-art non-buffered one-pass partitioning algorithms. Moreover, if the topology of a distributed system is known, it is possible to further optimize the communication costs by mapping partitions onto processing elements. Our experiments indicate that our algorithm is both faster and produces better process mappings than competing tools.  In case of graph partitioning,  our framework is up to two orders of magnitude faster at the cost of~5\% more cut edges compared to Fennel.
\end{abstract}

\section{Introduction}
\label{sec:introduction}

Graphs are a universal tool used to represent complex phenomena such as dependencies in databases, communications in distributed algorithms, interactions in social networks, and so forth. %
With the ever-increasing amount of data, processing large-scale graphs on distributed systems and databases becomes a necessity for a wide range of applications.
When processing a graph in parallel, each processing element (PE) operates on some portion of it while distinct PEs communicate with each other through message-passing.
To make the parallel algorithm communication-efficient, graph partitioning is needed as a preprocessing step. 
When the topology of the distributed system is known, a process mapping preprocessing is even more effective, since it directly minimizes the total communication cost for the given topology.
Moreover, network topologies of modern parallel computers 
feature special properties that can be exploited by mapping algorithms. In particular, a common
property that is frequently exploited by process mapping algorithms~\cite{DBLP:conf/europar/PredariTSM21,fonseca2020better} is that PEs are hierarchically organized into, e. g., islands, racks, nodes, GPUs,
processors, cores with corresponding communication links of similar quality.
Partitioning and process mapping have proven to significantly speedup parallel applications~\cite{aktulga2012topology,bhatele2009dynamic} and hierarchical partitionings by itself are used in a wide-range of applications, for example for distributed hybrid CPU and GPU training of graph neural networks on billion-scale graphs~\cite{zheng2022distributed}.

There has been a wide range of research on graph partitioning and process mapping~\cite{GPOverviewBook,SPPGPOverviewPaper,DBLP:reference/bdt/0003S19}.
Graph partitioning and process mapping are NP-complete \cite{Garey1974} and there can be no approximation algorithm with a constant ratio for general graphs~\cite{BuiJ92}. 
Thus, heuristics are used in practice. %
The most prominent results in the area can be split in three groups of algorithms: distributed and parallel, in-memory, and streaming.
From these groups, only streaming algorithms can quickly partition huge graphs while using much less memory than the total size of the graph.
This is why streaming partitioning has become a major research trend in the last years~\cite{stanton2012streaming,tsourakakis2014fennel,nishimura2013restreaming,awadelkarim2020prioritized,jafari2021fast,HeiStream}.
The most popular streaming approach in literature is the one-pass model~\cite{DBLP:journals/pvldb/AbbasKCV18}, where nodes arrive one at a time including their neighborhood and then have to be assigned to blocks. 
At the moment there is range of streaming algorithms operating in that model that are either very fast and don't care for solution quality at all (such as Hashing~\cite{stanton2012streaming}), or algorithms that are still fast, but compute significantly better solutions than just random assignments (such as such Fennel~\cite{tsourakakis2014fennel}). 
Moreover, to the best of our knowledge no streaming algorithm has yet been proposed or adapted to optimize for process~mapping~objectives.

\emph{Contribution.} 
We develop a shared-memory parallel streaming partitioning algorithm that performs recursive multi-sections on the fly thereby enabling the computation of hierarchical partitionings using a single pass over the input.
If a hierarchy is not specified as an input, out approach can also be used as a tool to solve the standard graph partitioning problem. 
Our approach has a considerably lower running time complexity in comparison with state-of-the-art non-buffered one-pass partitioning algorithms that solve this problem.
Moreover, it opens a door to trade solution quality for execution speed by solving some of its partitioning subproblems with faster algorithms.
Our experiments show that on average, our algorithm is $134.4$ times faster than Fennel at the cost of $5\%$ more cut edges.
In scalability tests, our algorithm is only $3$ times slower than Hashing when running on 32 threads.

Moreover, if the hierarchical topology of a distributed system is known, we can adopt our scheme to perform multi-sectioning along the specified hierarchy. Thereby our algorithms obtains good process mappings using single pass through the input graph. %
On average, our algorithm computes $41\%$ better process mappings and is 55 times faster than Fennel which ignores the given hierarchy.

\section{Preliminaries}
\label{sec:preliminaries}

\subsection{Basic Concepts}
\label{subsec:basic_concepts}

Let $G=(V=\{0,\ldots, n-1\},E)$ be an \emph{undirected graph} with no multiple or self edges, such that $n = |V|$, $m = |E|$.
Let $c: V \to \MdR_{\geq 0}$ be a node-weight function, and let $\omega: E \to \MdR_{>0}$ be an edge-weight function.
We generalize $c$ and $\omega$ functions to sets, such that $c(V') = \sum_{v\in V'}c(v)$ and $\omega(E') = \sum_{e\in E'}\omega(e)$.
Let $N(v) = \setGilt{u}{\set{v,u}\in E}$ denote the neighbors of $v$.
A graph $S=(V', E')$ is said to be a \emph{subgraph} of $G=(V, E)$ if $V' \subseteq V$ and $E' \subseteq E \cap (V' \times V')$. 
When $E' = E \cap (V' \times V')$, $S$ is an \emph{induced} subgraph.
Let $d(v)$ be the degree of node $v$ and $\Delta$ be the maximum degree of $G$.%

The \emph{graph partitioning} problem~(GP) consists of assigning each node of $G$ to exactly one of $k$ distinct \emph{blocks} respecting a balancing constraint in order to minimize the weight of the edges running between the blocks, i.e. the (edge-cut).
More precisely, GPP partitions $V$ into $k$ blocks $V_1$,\ldots,$V_k$ (\ie $V_1\cup\cdots\cup V_k=V$ and $V_i\cap V_j=\emptyset$ for $i\neq j$), which is called a \emph{\mbox{$k$-partition}} of $G$.
The \emph{edge-cut} of a $k$-partition consists of the total weight of the edges crossing blocks, \ie $\sum_{i<j}\omega(E_{ij})$, in which $E_{ij} \Is $ $\big\{\set{u,v}\in E : u\in V_i,v\in V_j\big\}$.
The \emph{balancing constraint} demands that the sum of node weights in each block does not exceed a threshold associated with some allowed \emph{imbalance}~$\epsilon$.
More specifically, $\forall i~\in~\{1,\ldots,k\} \gilt$ $c(V_i)\leq L_{\max}\Is \big\lceil(1+\epsilon) \frac{c(V)}{k} \big\rceil$.
Throughout this paper, a hierarchical partition is defined by $\mathcal{S}=a_1: a_2: ...:a_\ell$, were the total number of blocks is given as $k = \Pi_{i=1}^{\ell}a_i$ and $\mathcal{S}$ encodes that the partition is obtained by first partitioning the input into $a_\ell$ blocks, then each block is partitioned into $a_{\ell-1}$ blocks and so forth.

For process mapping applications of hierarchical partitions, assume that we have $n$ processes and a topology containing~$k$ PEs.
Let $\mathcal{C}\in \MdR^{n \times n}$ denote the communication matrix and let $\mathcal{D}\in \MdR^{k \times k}$ denote the (implicit) topology matrix or distance matrix.
In particular, $\mathcal{C}_{i,j}$ represents the required amount of communication between processes $i$ and $j$, while $\mathcal{D}_{x,y}$ represents the cost of each communication between PEs $x$ and $y$.
Hence, if processes $i$ and $j$ are respectively assigned to PEs $x$ and $y$, or vice-versa, the communication cost between $i$ and $j$ will be $\mathcal{C}_{i,j}\mathcal{D}_{x,y}$.
Throughout this work, we assume that $\mathcal{C}$ and $\mathcal{D}$ are symmetric -- otherwise one can create equivalent problems with symmetric inputs \cite{brandfass2013rank}.

In particular, for process mapping applications tackled in this paper, we assume that topologies are organized as homogeneous hierarchies. 
In this case $\mathcal{S}=a_1: a_2: ...:a_\ell$ is a
sequence describing the hierarchy of a supercomputer. The sequence
should be interpreted as each processor having $a_1$ cores, each node
$a_2$ processors, each rack $a_3$ nodes, and so~forth,
such that the total number of PEs is $k=\Pi_{i=1}^{\ell}a_i$.
Without loss of generality, we assume that $a_i \geq 2, \forall i \in \{1, \ldots, \ell\}$.
Let $D = d_1:d_2:\ldots:d_\ell$ be a sequence describing the distance between PEs within each hierarchy level, meaning that the distance between two cores in the same processor is $d_1$, the distance between two cores in the same node but in different processors is $d_2$, the distance between two cores in the same rack but in different nodes is $d_3$, and~so~forth.

Throughout the paper, we assume that the input communication matrix is already given as a graph $G_\mathcal{C}$, \ie no conversion of the matrix into a graph is necessary.
More precisely, the graph representation is defined as $G_\mathcal{C}:=(\{1,\ldots, n\}, E[\mathcal{C}])$ where $E[\mathcal{C}] :=\{(u,v) \mid \mathcal{C}_{u,v} \not = 0\}$.
In other words, $E[\mathcal{C}]$ is the edge set of the processes that need to communicate with each other. 
Note that the set contains forward and backward edges, and that the weight of each edge in the graph equals the corresponding entry in the~communication~matrix~$\mathcal{C}$. 

The \emph{process mapping} problem consists of assigning the nodes of a communication graph to PEs in a communication topology while respecting a balancing constraint (the same used for graph partitioning above) in order to minimize the total communication cost.
Let \mbox{$\Pi: \{1, \ldots, n\} \mapsto \{1, \ldots, k\}$} be the function that maps a node onto its PE.
The  objective of process mapping is to minimize $J(\mathcal{C},\mathcal{D}, \Pi) := \sum_{i,j} \mathcal{C}_{i, j}\mathcal{D}_{\Pi(i),\Pi(j)}$.
Within the scope of this work, the number of nodes (processes)~$n$ in the communication graph is much larger than the number of PEs~$k$, which matches most real-world situations.

\paragraph*{Streaming Model.}
\label{subsec:Computational Model}
Streaming algorithms usually follow a load-compute-store logic.
Our focus in this paper and the most used streaming model is the \emph{one-pass} model.
In this model, nodes are loaded one at a time alongside with their adjacency lists, then some logic is applied to permanently assign them to blocks.
This logic can be as simple as a Hashing function or as complex as scoring all blocks based on some objective and then assigning the node to the block with highest score.
There are other streaming models such as the sliding window~\cite{patwary2019window} and the buffered streaming~\cite{jafari2021fast,HeiStream}.
More details are given in Section~\ref{subsec:related_work}.

\vspace*{-.2cm}
\subsection{Related Work}
\label{subsec:related_work}

There is a huge body of research on graph partitioning.
The most prominent tools to partition (hyper)graphs in memory include Metis~\cite{karypis1996parallel}, Scotch~\cite{Pellegrini96experimentalanalysis}, KaHIP~\cite{kabapeE}, KaMinPar~\cite{gottesburen2021deep}, KaHyPar~\cite{schlag2016k}, Mt-KaHyPar~\cite{mt-kahypar-d}, and mt-KaHIP~\cite{DBLP:conf/europar/Akhremtsev0018}.
The readers are referred to~\cite{GPOverviewBook,SPPGPOverviewPaper,DBLP:reference/bdt/0003S19,walshaw2000mpm,karypis1998fast,Walshaw07,kaffpa} for extensive material and references.
Here, we focus on the results specifically related to the scope of this paper.

\paragraph*{Process Mapping.}
Müller-Merbach~\cite{muller2013optimale} proposed a greedy construction method to perform the one-to-one mapping of blocks onto PEs.
This method was later improved by Glantz~et~al.~\cite{glantzMapping2015} with an algorithm called \emph{GreedyAllC}. 
Heider~\cite{heider1972computationally} proposed a local search method to refine an already given mapping. %
To reduce the runtime, Brandfass~et~al.~\cite{brandfass2013rank} introduced a couple of modifications to speed up this local search, such as only considering pairs of PEs that can reduce the objective or partitioning the search space into some consecutive blocks and only performing swaps inside those blocks.
Glantz~et~al.~\cite{DBLP:conf/icpp/GlantzPM18} proposed a one-to-one mapping algorithm in which the hardware topology is an isometric subgraph of a hypercube and labeled the nodes and the PEs with bit-strings in order to optimize the algorithm locality.
Schulz~and~Träff~\cite{schulz2017better}, proposed a top-down multi-section approach to map blocks to PEs when the communication topology is a regular hierarchy. 
Then, the authors solved the process mapping problem in two steps by first partitioning the graph with KaHIP~\cite{kaffpa}, then applying its top-down algorithm followed by a local search similar to that from~\cite{brandfass2013rank}.
Subsequently, Kirchbach~et~al.~\cite{GlobalMultisection} modified the partitioning step to apply the KaHIP algorithm in a \emph{recursive multi-section} itself, which improved the solution quality even further.
We give more details about this type of algorithms in Section~\ref{sec:OnlineRecursiveMulti-Section}.
Faraj~et~al.~\cite{fonseca2020better} proposed an algorithm which integrates graph partitioning and process mapping, which considerably improved the mapping quality while also reducing the running time.
There are also specialized algorithms for process mapping, such as~\cite{DBLP:conf/cluster/KirchbachLH0T20}, which assumes a Cartesian communication topology.
There are publicly available tools which solve the process mapping problem, such as Scotch~\cite{Pellegrini96experimentalanalysis} and KaHIP~\cite{schulz2017better,GlobalMultisection}.
Moreover, Predari~et.al~\cite{DBLP:conf/europar/PredariTSM21}~released a distributed algorithm to solve the process mapping problem.
To the best of our knowledge, no work has yet solved process mapping using a streaming model.

\paragraph*{Streaming Partitioning.}
Stanton and Kliot~\cite{stanton2012streaming} introduced graph partitioning in the streaming model and proposed some heuristics to solve it.
Their most prominent heuristic include the one-pass methods \emph{Hashing} and \emph{linear deterministic greedy}~(LDG).
In their experiments, LDG had the best overall edge-cut.
In this algorithm, node assignments prioritize blocks containing more neighbors and use a penalty multiplier to control imbalance. 
Particularly, a node $v$ is assigned to the block $V_i$ that maximizes $|V_i \cap N(v)|*\lambda(i)$ with $\lambda(i)$ being a multiplicative penalty defined as $(1-\frac{|V_i|}{L_\text{max}})$. 
The intuition is that the penalty avoids to overload blocks that are already very heavy.
In case of ties on the objective function, LDG moves the node to the block with fewer nodes.
Overall, LDG partitions a graph in $O(m+nk)$ time.
On the other hand, Hashing has running time $O(n)$ but produces a poor edge-cut.

Tsourakakis~et~al.~\cite{tsourakakis2014fennel} proposed \emph{Fennel},  a one-pass partitioning heuristic based on the widely-known clustering objective \emph{modularity}~\cite{brandes2007modularity}.
Fennel assigns a node $v$ to a block $V_i$, respecting a balancing threshold, in order to maximize an expression of type $|V_i\cap N(v)|-f(|V_i|)$, \ie with an additive penalty.
This expression is an interpolation of two properties: attraction to blocks with many neighbors and repulsion from blocks with many non-neighbors.
When $f(|V_i|)$ is a constant, the expression coincides with the first property.
If $f(|V_i|) = |V_i|$, the expression coincides with the second property.
In particular, the authors defined the Fennel objective with $f(|V_i|) = \alpha * \gamma * |V_i|^{\gamma-1}$, in which~$\gamma$ is a free parameter and $\alpha = m \frac{k^{\gamma-1}}{n^{\gamma}}$.
After a parameter tuning made by the authors, Fennel uses $\gamma=\frac{3}{2}$, which provides $\alpha=\sqrt{k}\frac{m}{n^{3/2}}$.
As LDG, Fennel partitions a graph in $O(m+nk)$ time.

Restreaming graph partitioning has been introduced by Nishimura~and~Ugander~\cite{nishimura2013restreaming}.
In this model, multiple passes through the entire input are allowed, which enables iterative improvements.
The authors proposed easily implementable restreaming versions of LDG and Fennel: ReLDG and ReFennel, respectively. 
Awadelkarim and Ugander~\cite{awadelkarim2020prioritized} studied the effect of node ordering for streaming graph partitioning.
The authors introduced the notion of \emph{prioritized streaming}, in which (re)streamed nodes are statically or dynamically reordered based on some priority. 
The authors proposed a prioritized version of ReLDG, which uses multiplicative weights of restreaming algorithms and adapts the ordering of the streaming process inspired by balanced label propagation.  
Their experimental results consider a range of stream orders, where a dynamic ordering based on their own metric \emph{ambivalence} is the best regarding edge-cut, with a static ordering based on degree being nearly as good.

Besides the one-pass model, other streaming models have also been used to solve graph partitioning.
Patwary et al.~\cite{patwary2019window} proposed WStream, a greedy stream algorithm that keeps a sliding stream window.
Jafari et al.~\cite{jafari2021fast} proposed a shared-memory multilevel algorithm based on a buffered streaming model.
Their algorithm uses the one-pass algorithm LDG as the coarsening, initial partitioning, and the local search steps of their multilevel scheme.
Faraj~and~Schulz~\cite{HeiStream} proposed HeiStream, a high-quality multilevel algorithm also based on a buffered streaming model.
Their algorithm loads a chunk of nodes, builds a model, and then partitions this model with a traditional multilevel algorithm coupled with an extended version of the Fennel objective.
Regarding high-quality streaming partitioning, HeiStream is the only algorithm whose complexity is better than $O(m+nk)$.
In particular, it is $O(m+n)$, but it is slower than LDG in practice for reasonable values of $k$ due to larger constant factors~\cite{HeiStream}.
However, note that none of the streaming tools above solves the process mapping problem.

\vspace*{-.2cm}
\section{Online Recursive Multi-Section}
\label{sec:Online Recursive Multi-Section}
\label{sec:OnlineRecursiveMulti-Section}

In this section, we describe our main algorithmic contribution.
We describe the overall scheme and show how it can be used to compute recursive multi-sections on the fly along a specified \emph{hierarchy}. Moreover, we will also show how this applies to the process mapping application. Lastly, we show how to adapt the algorithm to solve standard the \emph{graph partitioning} problem.

\vspace*{-.2cm}
\subsection{Overall Scheme}
\label{subsec:Overall Scheme Map}

A successful offline algorithm to partition and map the nodes of a graph onto PEs is the recursive multi-section~\cite{schulz2017better,GlobalMultisection,fonseca2020better}.
This approach specializes the partitioning process for the case in which the communication cost between two processes (nodes) highly depends on the hierarchy level shared by the PEs (blocks) in which they are allocated.
Recall that a hierarchy is represented by a string $\mathcal{S}=a_1: a_2: ...:a_\ell$ which in the process mapping application means that each processor has $a_1$ cores, each node has
$a_2$ processors, each rack has $a_3$ nodes, and so forth.
The offline recursive multi-section works as follows.
First, the whole graph is partitioned in $a_\ell$ blocks. 
Then, the subgraph induced by each block of an $a_i$-partition is recursively partitioned in $a_{i-1}$ sub-blocks until the whole graph is partitioned in $k=\prod_{i=1}^{\ell}{a_i}$ blocks.
This recursive approach have been developed for process mapping and exploits the fact that the communication between PEs is cheaper through lower layers of the communication hierarchy.
It creates a hierarchy of partitioning sub-problems that directly reflects the hierarchical topology of the system. 
This yields an improved process mapping in practice~\cite{GlobalMultisection,fonseca2020better}.

Intuitively, for any given hierarchy $S$ (independent of the application) recursive multi-section can be implemented in the streaming model with $\ell$ successive passes of any one-pass partitioning algorithm over the input graph.
In the first pass, the whole graph is partitioned in $a_\ell$ blocks using a scoring function like  Fennel or LDG.
In the second pass, the nodes previously assigned to each block are partitioned in $a_{\ell-1}$ sub-blocks.
The same logic is propagated until the $\ell^{th}$ pass, when each node is finally assigned to a unique PE.
Since there are no assignments of nodes until the last pass over the graph, this is not an online algorithm.
We make this algorithm online by compressing all steps performed during the $\ell$ passes in a single pass, as follows.
After a node is loaded, assign it to one of the $a_\ell$ blocks $\{V^\ell_1,\ldots,V^\ell_{a_\ell}\}$ in layer~$\ell$.
Then, for each layer $i<\ell$, assign the node to one of the $a_{i}$ sub-blocks of the block chosen in the previous step.
After going through all layers, the node is directly assigned to a PE, which makes this approach feasible for online execution.
Algorithm~\ref{alg:One-pass Process Mapping} summarizes the structure of our online recursive multi-section. Here, score depends on the algorithm logic used, e.g. Fennel or LDG, as well as other parameters that are specific to our multi-section algorithm. We give more details in Section \ref{subsec:subproblems}.
Figure~\ref{fig:hierarchy_problems} illustrates how a node is assigned to a permanent block on the fly using our online recursive multi-section.
Note that it produces exactly the same result as the version with $\ell$ passes since it does not violate any data dependency.
This is the case since each decision made in a given pass of the multiple-pass version only relies on nodes previously streamed during that same pass.

\begin{algorithm}[t] %
	\caption{Online Recursive Multi-Section} %
	\label{alg:One-pass Process Mapping} %
	\begin{algorithmic}[1] %
		\For{$u \in V(G)$}
		\State Load $N(u)$
		\State $X \leftarrow \{V^\ell_1,\ldots,V^\ell_{a_\ell}\}$, $V^* \leftarrow \emptyset$
		
		\For{$i \in \{\ell,\ldots,1\}$}
		\State \textbf{if} $i \neq \ell$ \textbf{then} $X \leftarrow $ sub-blocks of $V^*$
		\State $V^* \leftarrow$ $\arg \max\limits_{W\in X}\big\{score(W)\big\}$
		\State Assign vertex to sub-block $V^*$
		\EndFor
		\EndFor
	\end{algorithmic}
\end{algorithm}

\begin{figure}[b!]
	\centering
	\includegraphics[width=0.4\linewidth]{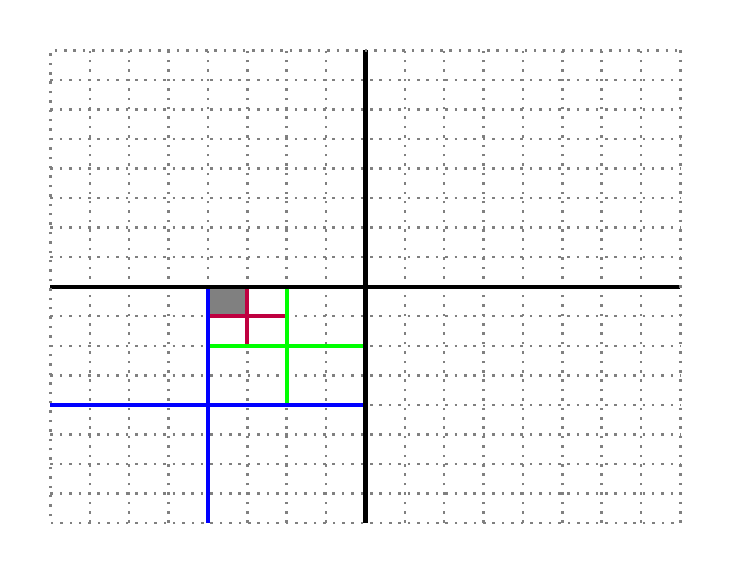}
	\caption{Assigning a node over a grid of ${256}$ blocks contained in a hierarchy ${S=4:4:4:4}$ with online recursive multi-section. In the first ${4}$-partition (black), the node is assigned to the lower left block. In the following ${4}$-partitions, it is assigned to the upper right sub-block (blue), its upper left sub-block (green), and finally its upper left sub-block (purple), which is a block from the original partitioning problem.}
	\label{fig:hierarchy_problems}
	\vspace*{-.5cm}
\end{figure}

Note that the algorithm can be applied to compute any hierarchical partitioning and is not limited to the process mapping application. 
However, as the offline multi-section, in case of the process mapping application our online multi-section exploits the inherent structure of the hierarchical process mapping problem in two ways:
(i) its layout reproduces the hierarchical communication topology;
(ii) the top-down order in which the nodes are assigned to sub-blocks of previously assigned blocks reflects the order in which the communication costs decrease in a communication topology.
In other words, this top-down order reflects the need for primarily avoiding cut edges among modules of higher layers in the communication hierarchy.
We now analyze the time and space complexity of our algorithm.
When implementing the online multi-section algorithm, we need to keep the weight of the complete hierarchy of blocks and sub-blocks in memory.
That is necessary since one-pass partitioning algorithms such as Fennel and LDG keep track of block weights in order to compute their scores.
Following from this fact, we show the space complexity of our algorithm in Lemma~\ref{theo:num_blocks_msec} and Theorem~\ref{theo:mem_usage_msec}.
Then, we derive the running time complexity of our algorithm in Theorem~\ref{theo:running_time_msec} and show an important special case in Corollary~\ref{theo:expected_running_time_msec}.

\begin{lemma}
	\label{theo:num_blocks_msec}
	Online recursive multi-section needs $O(k)$ space to store block weights throughout the algorithm.
\end{lemma}
\vspace*{-.3cm}
\begin{proof}
	By definition, the multi-section consists of $\ell$~layers such that a layer $i \in \{1,\ldots,\ell\}$ contains exactly $\prod_{r=i}^{\ell}{a_r}$~blocks whose weight we need to keep track of.
	As we define $a_r \geq 2 \ \forall r$, we can write $\prod_{r=i}^{\ell}{a_r} \leq (1/{2^{i-1}})\prod_{r=1}^{\ell}{a_r} = (1/{2^{i-1}})k$.
	Hence, the total amount of block weights that we need to keep track of is $\sum_{i=1}^{\ell}{\prod_{r=i}^{\ell}{a_r}} \leq \sum_{i=1}^{\ell}{(1/{2^{i-1}})k} \leq 2k$.
\end{proof}

\vspace*{-.2cm}
\begin{theorem}
	\label{theo:mem_usage_msec}
	Online recursive multi-section coupled with Fennel or LDG needs $O(n+k)$ memory.
\end{theorem}
\vspace*{-.3cm}
\begin{proof}
	Due to the hierarchical structure of multi-section, Fennel and LDG may keep track of a single block assignment per node, which is enough to infer all its superblocks.
	Hence, the space complexity $O(n+k)$ directly follows from Lemma \ref{theo:num_blocks_msec}.
\end{proof}
\vspace*{-.2cm}
\begin{theorem}
	\label{theo:running_time_msec}
	Online recursive multi-section coupled with Fennel or LDG has time complexity $O(m\ell + n\sum_{i=1}^{\ell}{a_i})$.
\end{theorem}
\vspace*{-.3cm}
\begin{proof}
	The online recursive multi-section assigns each node~$u$ over $\ell$ layers.
	Using Fennel or LDG, the running time to assign~$u$ in a given layer~$i$ is $O(|N(u)|+a_i)$.
	Accounting for all layers and nodes, this sums up to $O(m\ell + n\sum_{i=1}^{\ell}{a_i})$.
\end{proof}
\vspace*{-.2cm}
\begin{corollary}
	\label{theo:expected_running_time_msec}
	Given a constant $b \geq 2$, if $a_i = b$, $\forall i \in \{1,\ldots,\ell\}$, then online recursive multi-section coupled with Fennel or LDG has time complexity $O\big((m + n)\log{k}\big)$.
\end{corollary}
\vspace*{-.3cm}
\begin{proof}
	Based on the assumption, we derive the claimed bound from Theorem~\ref{theo:running_time_msec} by proving $\ell=O(\log{k})$ and $\sum_{i=1}^{\ell}{a_i} = O(\log{k})$.
	The first part trivially holds since $k=\prod_{i=1}^{\ell}{a_i}=b^\ell \Rightarrow \ell = \log_{b}{k}$.
	To prove the second part, notice that $\sum_{i=1}^{\ell}{a_i} = b\ell$.
	Since $\ell = \log_{b}{k}$, it follows that $k=\prod_{i=1}^{\ell}{a_i} = b\log_{b}{k} = O(\log{k})$.
\end{proof}

\subsection{Partitioning Subproblems}
\label{subsec:Process Mapping}
\label{subsec:subproblems}

Our multi-section algorithm implies a hierarchy of one-pass partitioning subproblems.
These subproblems are self contained one-pass partitioning problems, so they can be solved by any one-pass partitioning algorithm in literature.
In this section, we examine these subproblems based on the parameters of the some process mapping problem. 
Consider all the partitioning subproblems contained in some layer $i$ of our multi-section.
First, note that they are homogeneous, which means that they all receive as input an induced subgraph with roughly the same number of nodes and edges and partition it among $a_i$ blocks.
More specifically, a subproblem in layer $i$ partitions among $k_i = a_{i}$ blocks and receives as input a graph containing roughly $n_i = n/{\prod}_{r=i+1}^{\ell}{a_r}$ nodes and $m_i = m/{\prod}_{r=i+1}^{\ell}{a_r}$ edges.
As a consequence, the size constraint $L_i$ of a block from subproblems in layer $i$ is $L_i =  \lceil(1+\epsilon) n_i/k_i \rceil 
\simeq L_{\max} \prod_{r=1}^{i-1}{a_r}$, which is simply the sum of capacities of all blocks from the original problem contained in it.
These variations in subproblem size have further implications depending on the partitioning algorithm coupled with our scheme, as we show next.

\paragraph*{Fennel Mapping.}
Using Fennel in the online recursive multi-section requires attention to its constant $\alpha$.
Recall that it is defined as $\alpha = \sqrt{k} m/ n^{3/2}$ for partitioning the whole graph into $k$ blocks with vanilla Fennel.
Using this value of $\alpha$ for all partitioning subproblems contained in our multi-section is not a natural choice since we intend to apply Fennel independently for each subproblem.
Independently applying the definition of Fennel for each partitioning subproblem contained in our multi-section implies~$\ell$ different parameters $\alpha_i$, $i \in \{1, \ldots, \ell\}$, for all multi-section layers.
We derive the value of~$\alpha_i$ by applying the Fennel definition $\alpha_i = \sqrt{k_i} \frac{m_i}{ n_i^{3/2}}$ and substituting the values of $k_i$, $m_i$, and $n_i$ which we have already discussed.
It follows  that 
$\alpha_i = \frac{\alpha }{ \sqrt{\prod_{r=1}^{i-1}{a_r}}}$.

\paragraph*{LDG Mapping.}
Combining LDG with our multi-section is straightforward, since it directly uses the remaining capacity of each block as a multiplicative~penalty.
Hence, we can configure LDG for a subproblem in layer $i$ by simply computing this penalty based on the block capacity $L_i$, whose value we have already discussed.

\paragraph*{Hybrid  Mapping}
It is also possible to solve distinct subproblems with different partitioning algorithms.
This possibility opens a door to a trade-off when we mix a high-quality algorithm such as Fennel with a fast algorithm such as Hashing.
In particular, we can use Fennel to solve top-layer subproblems (whose communication is more expensive) and Hashing to solve bottom-layer subproblems (whose communication is cheaper).
When we solve $h$ ($1 \leq h \leq \ell$) upper layers of subproblems Fennel and the remaining ones with Hashing, we produce an algorithm with the running time given in Theorem~\ref{theo:hybrid_map}.
Note that this hybridization is faster than coupling the multi-section with Fennel only and slower than coupling it with Hashing only.

\begin{theorem}
	\label{theo:hybrid_map}
	Solving the $h$ upper layers of the	online recursive multi-section with Fennel and the remaining ones with Hashing has overall time complexity $O\big(mh + n\sum_{i=\ell-h+1}^{\ell}{a_i}\big)$.
\end{theorem}

\begin{proof}
	From Theorem~\ref{theo:running_time_msec}, running the top $h$ layers of the online multi-section with Fennel costs $O\big(mh + n\sum_{i=\ell-h+1}^{\ell}{a_i}\big)$. 
	Asymptotically, this complexity is not changed when we include the complexity of solving the remaining $(\ell-h)$ layers with Hashing.
	This is the case since it suffices to apply a single Hashing operation per node, which costs~$O(1)$.
\end{proof}

\paragraph*{Remapping.}
It is possible to iteratively improve a process mapping solution through multiple passes of our online multi-section algorithm over the input graph.
This can be achieved by coupling our algorithm with restreaming algorithms such as  ReFennel or ReLDG with the proper adaptations. However, remapping is beyond the scope of this paper.

\vspace*{-.1cm}
\subsection{General Partitioning}
\label{subsec:General Partitioning}

The previous sections assume that some hierarchy is given.
This hierarchy is then used as input for the online recursive multi-section.
In this section, we show how to partition a streamed graph into an arbitrary amount of blocks using the online recursive multi-section when no explicit hierarchy is given. 
We do this by creating an artificial hierarchy.
We derive the complexity of the proposed approach and discuss its~partitioning~subproblems.

The \emph{recursive bisection} is a successful offline approach to partition graphs into an arbitrary number $k$ of blocks. %
If~$k$ is a power of $2$, the algorithm works as a recursive multi-section with $\log_2{k}$ layers of $2$-way partitioning subproblems. 
Otherwise, it is irregular and cannot be represented by a string~$S$.
Analogously, we define an online recursive bisection to partition a graph on the fly when no hierarchy is given.
Recall that the whole hierarchy of blocks and sub-blocks has to be kept in memory throughout the execution of the online recursive multi-section, hence the same requirement applies here.
We build this hierarchy, which we call \emph{multi-section tree}, as a preliminary step for the streaming partitioning process.
In Algorithm~\ref{alg:Recursively Define Subproblems}, we define the procedure $\textsc{BuildHierarchy}$ which recursively builds this multi-section tree for any value~of~$k$.
This procedure receives as input a parent block $P$ for the multi-section tree as well as the endpoints $k_L$ and $k_R$ of the range of blocks to be covered by the multi-section tree.
In line~2, it terminates the recursion when $P$ is a leaf of the multi-section tree, which is true when $k_L=k_R$.
Otherwise, it creates two sub-blocks for $P$ and inserts them as sons of $P$ in the multi-section tree (line~3). 
Then, it splits the range $\{k_L,\ldots,k_R\}$ in roughly equal parts and performs 2 recursive calls to itself.

\begin{algorithm} %
	\caption{Create Blocks for Multi-Section Tree} 
	\label{alg:Recursively Define Subproblems} %
	\raggedright
	\hspace*{\algorithmicindent} \textbf{Input} \\
	\hspace*{\algorithmicindent} $P$: Parent block in the hierarchy \\
	\hspace*{\algorithmicindent} $k_L$: Left endpoint of blocks covered by hierarchy \\
	\hspace*{\algorithmicindent} $k_R$: Right endpoint of blocks covered by hierarchy \\
	\begin{algorithmic}[1] %
		\Procedure{BuildHierarchy}{$P, k_L, k_R$}
		
		\If{$k_L = k_R$}
		\Return \EndIf
		\State $P_L,P_R \leftarrow$ Create sub-blocks for $P$
		\State $\textsc{BuildHierarchy}\Big(P_L, k_L, \big\lfloor \frac{k_L+k_R}{2} \big\rfloor\Big)$
		\State $\textsc{BuildHierarchy}\Big(P_R, \big\lfloor \frac{k_L+k_R}{2} \big\rfloor+1, k_R\Big)$
		\EndProcedure
	\end{algorithmic}
\end{algorithm}

We further generalize the recursive bisection to \emph{recursive $b$-section} for a \emph{base}~$b$.
Given a base $b \geq 2$, a recursive $b$-section is a recursive multi-section associated with a multi-section tree in which blocks have up to $b$ sub-blocks.
Algorithm~\ref{alg:Recursively Define Subproblems} can be adapted to deal with~$b$-section by creating $\min\{b,k_R-k_L+1\}$ sub-blocks in line~3 and, afterwards, making the same amount of recursive calls with proper parameters.

We create the multi-section tree by calling the command $\textsc{BuildHierarchy}(\emptyset,1,k)$ at the cost of $O(k)$.
Given a multi-section tree, we solve it by using Algorithm~\ref{alg:One-pass Process Mapping}.
Analogously to Lemma~\ref{theo:num_blocks_msec} and Theorem~\ref{theo:mem_usage_msec}, it is possible to prove that the online recursive $b$-section respectively stores $O(k)$ blocks and needs $O(n+k)$ memory when coupled with Fennel or LDG.
Theorem~\ref{theo:expected_running_time_bisec} provides a running time bound.

\begin{theorem}
	\label{theo:expected_running_time_bisec}
	Online recursive $b$-section coupled with Fennel or LDG has time complexity $O\big( (m+nb) \log_b{k}\big)$.
\end{theorem}

\begin{proof}
	The number of layers in the multi-section tree is up to $\lceil\log_b{k}\rceil$.
	In other words, each node should be assigned through up to $1+\log_b{k}$ layers.
	Since all subproblems partition among up to $b$~blocks, then the running time to assign a node $u$ over a layer is $|N(u)|+b$.
	Accounting for all layers and nodes, this sums up to $(2m + nb)(\log_b{k}+1) = O\big( (m+nb) \log_b{k}\big)$.
\end{proof}

\newcommand{ \scaleFactorSmall} {0.30}
\newcommand{ \imgScaleFactorSmall} {0.9}
\newcommand{ \capPositionSmall} {-.45cm}
\newcommand{ \afterCapSmall} {-.45cm}

\paragraph*{Heterogeneous Partitioning.}
\label{subsec:Heterogeneous Partitioning}

When $k$ is not a power of $b$, the recursive $b$-section hierarchy may contain some partitioning subproblems with heterogeneous blocks.
We deal with this by computing the size constraint of each block in the multi-section tree individually.
For simplicity, we explain how to do this when $b=2$ (recursive bisection), but this can be easily extended to an arbitrary~$b$.
For example when $k=5$, the two blocks in the first 2-way partitioning subproblem respectively cover $2$ and $3$ of the blocks from the original $5$-way partitioning.
Hence, these two blocks shall respectively have capacities $2L_{\max}$ and $3L_{\max}$, where $L_{\max}$ is the size constraint of a block in the original $k$-way partitioning.
Putting it in general terms, each block from the multi-section tree created in line $3$ of Algorithm~\ref{alg:Recursively Define Subproblems} covers $t = k_R - k_L + 1$ blocks of the original $k$-way partitioning.
For simplicity, we use $t$ to refer to this amount covered by a given block, and we use $t_1$ and $t_2$ to refer to the amounts covered by the two blocks of a partitioning subproblem.
The size constraint of a block is $t\times L_{\max}$.

When a subproblem has blocks with heterogeneous size constraints, the used partitioning algorithm has to cope with it.
We adapt \emph{Fennel} to address this issue by increasing (decreasing) the constant $\alpha$ used to compute the score of a specific block when its size constraint is lower (higher) than the other block from the same subproblem.
Recall that $\alpha$ depends on the numbers of nodes, edges and blocks for a specific subproblem.
A subproblem receives as input an induced subgraph with roughly $\frac{t_1+t_2}{k}$ of the nodes and edges from the original $k$-way partitioning.
We redefine the number of blocks as $\frac{t_1+t_2}{t_1}$ for the first block and $\frac{t_1+t_2}{t_2}$ for the second block of a subproblem.
This value equals $2$ for both blocks when $t_1=t_2$. 
Nevertheless, if $t_1 \neq t_2$, this value is larger (smaller) than $2$ for the block with smaller (larger) size constraint.
Summing up, the value of $\alpha$ for a given block will be $\sqrt{t}$ times smaller than the value $\alpha$ from the original $k$-way partitioning problem.
Consequently, the Fennel penalty function for imbalance will be more weighted for blocks with lower capacity, which tackles the heterogeneous balancing issue.
For \emph{LDG}, a natural adaptation for heterogeneous blocks arises from its very definition, since it directly uses the remaining capacity of each block as a multiplicative~penalty.

{}

\subsection{Shared-Memory Parallelization}
\label{subsec:Shared-Memory Parallelization}

Since the online recursive multi-section is a vertex-centric algorithm, we can parallelize it by independently splitting the nodes of the graph among threads.
More specifically, it can be achieved with OpenMP by parallelizing the \emph{for} loop in line 1 of Algorithm~\ref{alg:One-pass Process Mapping}.
This parallelization requires the nodes from the input graph to be concurrently loaded by distinct threads alongside with their neighborhoods, which is a reasonable assumption in many practical environments. 
Regarding data consistency, the only source of concern are the block weights, whose values can be concurrently read and incremented by multiple threads.
This is important because an inconsistency could compromise the load balance between blocks.
We ensure writing consistency by making the incrementation an atomic operation.
Potentially, a block can still be overloaded if multiple threads decide to assign a node to it at the same time.
Since this is very unlikely, we do not use any synchronization  to keep it from happening.
Finally, there is no relevant concern about the consistency of assigning nodes to blocks, since it is written by a single thread and can be read by multiple threads.

{}

\section{Experimental Evaluation}
\label{sec:Experimental Evaluation}

\paragraph*{Methodology.} 
We performed our implementations inside the KaHIP framework (using C++) and compiled them using gcc 9.3 with full optimization turned on (-O3 flag). 
Since no official versions of Fennel, LDG, and Hashing are available in public repositories, we implemented them in our framework.
Our implementations of these algorithms reproduce the results presented in the respective papers and are optimized for running time as much as possible. 
We have used a machine with one sixteen-core Intel Xeon  Silver 4216 processor running at $2.1$ GHz, $100$ GB of main memory, $16$ MB of L2-Cache, and $22$ MB of L3-Cache running Ubuntu 20.04.1. The machine can handle 32 threads with hyperthreading. %
Unless otherwise mentioned we stream the input directly from the internal memory to obtain clear running time comparisons. However, note that the algorithm could also be run streaming the graph from hard~disk.
We perform two types of experiments: experiments for the process mapping objective (with given hierarchies as specified below) and standard graph partitioning.
Unless otherwise mentioned, we use the following configurations for process mapping experiments: $D=1:10:100$, $\mathcal{S}=4:16:r$, with $r \in \{1,2,3,\ldots,128\}$. Hence, $k=64 r$.
This is the same configuration used in previous studies~\cite{schulz2017better,GlobalMultisection,fonseca2020better}.
Analogously, we use $k = 64s$, $s \in \{1,2,3,\ldots,128\}$ for partitioning experiments unless mentioned otherwise.
We allow a fixed imbalance of $3\%$ for all experiments (and all algorithms) since this is a frequently used value in the partitioning literature. 
All partitions computed by all algorithms were balanced.
Depending on the focus of the experiment, we measure running time, edge-cut, and/or the mapping communication cost $J$.
We perform ten repetitions per algorithm and instance using random seeds for initialization, and we compute the arithmetic average of the computed objective functions and running time per instance.
When further averaging over multiple instances, we use the geometric mean in order to give every instance the same influence on the \textit{final score}. 
Unless mentioned otherwise, we average all results of each algorithm grouped by~$k$.
Given a result of an algorithm~$A$, we express its value $\sigma_A$ (which can be objective or running time) as \emph{improvement} over an algorithm~$B$, computed as $\big(\frac{\sigma_B}{\sigma_A}-1\big)*100\%$;
{}
We also present \emph{performance profiles} which relate the running time (quality) of a group of algorithms to the fastest (best) one on a per-instance basis (rather than grouped by $k$).
The x-axis shows a factor $\tau$ while the y-axis shows the percentage of instances for which A has up to $\tau$ times the running time (quality) of the fastest (best)~algorithm.

\paragraph*{Instances.}
We selected graphs from various sources to test our algorithm.
Most of the considered graphs were used for benchmark in previous works on graph partitioning.
The graphs wiki-Talk and web-Google, as well as most networks of co-purchasing, roads, social, web, autonomous systems, citations, circuits, similarity, meshes, and miscellaneous are publicly available either in~\cite{snapnets}.
Prior to our experiments, we converted these graphs to a vertex-stream format while removing parallel edges, self loops, and directions, and assigning unitary weight to all nodes and edges.
We also use graphs such as eu-2005 and in-2004, which are available at the 10$^{th}$ DIMACS Implementation Challenge website~\cite{benchmarksfornetworksanalysis}. 
Finally, we include some artificial random graphs.
We use the name \Id{rggX} for \emph{random geometric graph} with
$2^{X}$ nodes where nodes represent random points in the unit square and edges connect nodes whose Euclidean distance is below $0.55 \sqrt{ \ln n / n }$. 
We use the name \Id{delX} for a graph based on a Delaunay triangulation of $2^{X}$ random points in the unit square~\cite{kappa}.
Basic properties of the graphs under consideration can be found in Table~\ref{tab:graphs}.
For our experiments, we split the graphs in two disjoint sets.
In all experiments, we stream the graphs with the natural given order of the nodes.

\begin{table}[t]
	\centering
	\footnotesize
	\begin{tabular}{| l  r  r  r | }
		\hline
		Graph & $n$& $m$ & Type\\

		\hline

		\multicolumn{4}{|c|}{Test Set} \\
		\hline

		Dubcova1 & \numprint{16129} & \numprint{118440} & Meshes \\
		hcircuit & \numprint{105676}  & \numprint{203734} & Circuit \\

		coAuthorsDBLP & \numprint{299067}     & \numprint{977676}  & Citations \\
		Web-NotreDame & \numprint{325729}     & \numprint{1090108}  & Web \\
		Dblp-2010 & \numprint{326186}     & \numprint{807700}  & Citations \\
		ML\_Laplace & \numprint{377002} & \numprint{13656485} & Meshes \\
		coPapersCiteseer & \numprint{434102}     & \numprint{16036720}  & Citations \\
		coPapersDBLP & \numprint{540486}     & \numprint{15245729}  & Citations \\
		
		Amazon-2008 & \numprint{735323}  & \numprint{3523472} & Similarity \\
		eu-2005 & \numprint{862664}    & \numprint{16138468}  & Web \\
		web-Google & \numprint{916428}    & \numprint{4322051}  & Web \\

		ca-hollywood-2009 & \numprint{1087562} & \numprint{1541514} & Roads \\
		
		Flan\_1565 & \numprint{1564794} & \numprint{57920625} & Meshes \\
		
		Ljournal-2008 & \numprint{1957027}  & \numprint{2760388} & Social \\
		
		HV15R & \numprint{2017169}  & \numprint{162357569} & Meshes \\
		Bump\_2911 & \numprint{2911419}  & \numprint{62409240} & Meshes \\
		del21 & \numprint{2097152}  & \numprint{6291408} & Artificial \\	
		rgg21 & \numprint{2097152} & \numprint{14487995} & Artificial \\
		
		FullChip & \numprint{2987012} & \numprint{11817567} & Circuit \\
		soc-orkut-dir & \numprint{3072441} & \numprint{117185083} & Social \\
		patents & \numprint{3750822}     & \numprint{14970766}  & Citations \\
		cit-Patents & \numprint{3774768}     & \numprint{16518947}  & Citations \\
		soc-LiveJournal1 & \numprint{4847571} & \numprint{42851237}   & Social \\
		circuit5M & \numprint{5558326} & \numprint{26983926} & Circuit \\
		italy-osm & \numprint{6686493}  & \numprint{7013978} & Roads \\
		great-britain-osm & \numprint{7733822} & \numprint{8156517} & Roads \\

		\hline

	\end{tabular}
	\vspace*{.25cm}
	\caption{Graphs for experiments.}
	
	\label{tab:graphs}
	\vspace*{-0.75cm}
\end{table}

{}

{}

\paragraph*{Parameter Tuning.}
We performed extensive tuning experiments using the graphs disjoint from the graphs in Table~\ref{tab:graphs}. Instead, we briefly summarize the main results.
The online multi-section produces on average $3.89\%$ better mapping and $0.19\%$ better edge-cut when coupled with Fennel than when coupled with LDG. Hence, we use Fennel as our scoring function. Computing adapted values of $\alpha$ for each partitioning subproblem is superior than using the default value of $\alpha$ of the original $k$-way partitioning. 
Particularly, it is on average $3.1\%$ faster while producing $9.7\%$ better mapping and cutting roughly the same amount of edges. 
Hence, our algorithm uses \emph{adapted} $\alpha$ values.
When no communication hierarchy is given, using the \emph{base} $b=4$ to build the multi-section tree is the fastest configuration overall.  
Using $b=4$, our algorithm is $16.7\%$ faster and cuts $3.2\%$ fewer edges than using $b=2$.
Hence, our algorithm uses \emph{base} $b=4$.
Using hashing on lower levels of the multi-section tree increases has a higher impact on the edge cut than on the process mapping objective. 
As expected, running time is also additionally decreased in both cases. 
For example, solving $67\%$ of the bottom layers of the multi-section with Hashing produces the following results on average in comparison to the non-hybrid configuration: $2.3$ times more cut edges, $27.5\%$ higher mapping objective, and $31.1\%$ less running time compared to vanilla streaming multisection.
Hence, we conclude that the approach yields a nice trade-off between solution quality and running time, but do not perform further experiments.
From now on, we refer to our online multi-section algorithm as \emph{OMS} when a communication hierarchy is given and \emph{nh-OMS} otherwise.

\vspace*{-.1cm}
\subsection{State-of-the-Art}
\label{subsec:State-of-the-Art}

In this section, we show experiments in which we compare the online recursive multi-section against the current state-of-the-art.
Except when mentioned otherwise, these experiments involve all the graphs from Table~\ref{tab:graphs}.
We identify Fennel and LDG as the state-of-the-art of non-buffered one-pass stream partitioning algorithms which aim at minimizing edge-cut.
Since Fennel generates better solutions on average than LDG~\cite{tsourakakis2014fennel}, we focus our experiments on Fennel without loss of generality.
We also include Hashing as a competitor, since it is the fastest algorithm for streaming partitioning.
To the best of our knowledge, there are no streaming partitioning algorithms specifically designed for process mapping.
For comparison purposes, we ran experiments with internal-memory tools, i.e. we compare to the fastest version of the integrated multi-level algorithm proposed in~\cite{fonseca2020better}, which we refer to as~\emph{IntMap}.
In addition, we compare our results against KaMinPar~\cite{gottesburen2021deep}, a very fast internal-memory algorithm that is orders of magnitudes faster than mt-Metis in terms of running time and produces comparable cuts while also enforcing balance (in contrast to mt-Metis). In particular, the purpose of runing IntMap and KaMinPar is to provide a reference of streaming algorithms in comparision to internal memory algorithms. 
We set a timeout of $30$ minutes for an algorithm to partition a graph, after which the execution is interrupted.
The only algorithm which exceeded this time limit for some instances was IntMap.
Hence, we exclude this algorithm from the plots.

\paragraph*{Solution Quality (Process Mapping).}
We start by looking at the mapping quality produced by OMS. 
In Figure\ref{fig:stateoftheartMap_res}, we plot the average mapping improvement over Hashing. 
KaMinPar produces the best mapping overall, with an average improvement of $1117\%$ over Hashing.
Among the instances IntMap could solve, it improves on average $7.6\%$ over KaMinPar. 
IntMap produces the best overall mapping in $67\%$ of the cases it could solve. 
Note that this is in line with previous works in the area of graph partitioning, i.e. streaming algorithms typically compute worse solutions than internal memory algorithms that always have access to the whole graph. 
OMS has an average improvement of $257.8\%$ over Hashing, while Fennel improves $153\%$ on average over Hashing.
In a direct comparison, OMS produces on average $41\%$ better mappings than Fennel.
In Figure~\ref{fig:stateoftheartMap_respp}, we plot the mapping performance profile.
In the plot, KaMinPar produces the best overall mapping for all instances instances.
We conclude that OMS produces the best mapping among the streaming~competitors.

\paragraph*{Solution Quality (Edge-Cut).}

Next we look at the edge-cut of nh-OMS.
In Figure~\ref{fig:stateoftheartPar_res}, we plot the edge-cut improvement over Hashing.
KaMinPar produces the best overall edge-cut, with an average improvement of $3024\%$ over Hashing.
IntMap cuts $20\%$ more edges on average than KaMinPar for the instances it solved. 
Among the streaming algorithms, Fennel and nh-OMS produce improvements of respectively $130.5\%$ and $118.2\%$ on average over Hashing.
In a direct comparison, nh-OMS cuts on average $5\%$ more edges than Fennel.
In Figure~\ref{fig:stateoftheartPar_GPpp}, we plot the edge-cut performance profile.
KaMinPar produces the smallest edge-cut for all instances.
Among the streaming algorithms, Fennel is slightly better than nh-OMS and both are distinctly better than Hashing.

\paragraph*{Running Time.}We now investigate the running time of OMS and nh-OMS.
In Figure~\ref{fig:stateoftheartPar_tim}, we plot the speedup over Fennel.
On average, Hashing is~$1301$~times faster than Fennel, while nh-OMS and OMS are respectively $133$~and~$55.4$ times faster than Fennel.
In a direct comparison, Hashing is on average $9.7$ times faster than nh-OMS and $23.4$ times faster than OMS.
KaMinPar comes next with an average speedup of $5.3$ over Fennel.
In a direct comparison, nh-OMS and OMS are respectively $25.1$ and $10.5$ times faster than KaMinPar.
KaMinPar is on average $2.5$ times faster than IntMap for the instances IntMap could solve. 
In Figure~\ref{fig:stateoftheartPar_timpp}, we plot the running time performance profile.
Note that the running time of nh-OMS is at most 16 times slower than Hashing for $100\%$ of the experiments, which is in accordance with Theorem~\ref{theo:expected_running_time_bisec}.
As the third fastest algorithm, OMS is considerably faster than all the other competitors, including Fennel.

\paragraph*{Memory Requirements.}We now look at the memory requirements of the different algorithms. We measured this on three graphs of our collection. In this case, we run stream graphs directly from disk for the streaming algorithms.
Besides being the fastest algorithm, Hashing needs the least memory overall.
For soc-orkut-dir, HV15R, and soc-LiveJournal1, it respectively needs $17$MB, $13$MB, and $24$MB.
OMS, nh-OMS, and Fennel have comparable consumption, all of which using $19$MB, $14$MB, and $25$MB for the mentioned graphs, respectively.
Finally, KaMinPar respectively uses $4.1$GB, $4.1$GB, and $1.8$GB, while IntMap respectively uses $34$GB, $12$GB, and $10$GB.

\begin{figure*}[t]
	\captionsetup[subfigure]{justification=centering}
	\centering
	\begin{subfigure}[]{\scaleFactorSmall\textwidth}
		\centering
		\includegraphics[width=\imgScaleFactorSmall\textwidth]{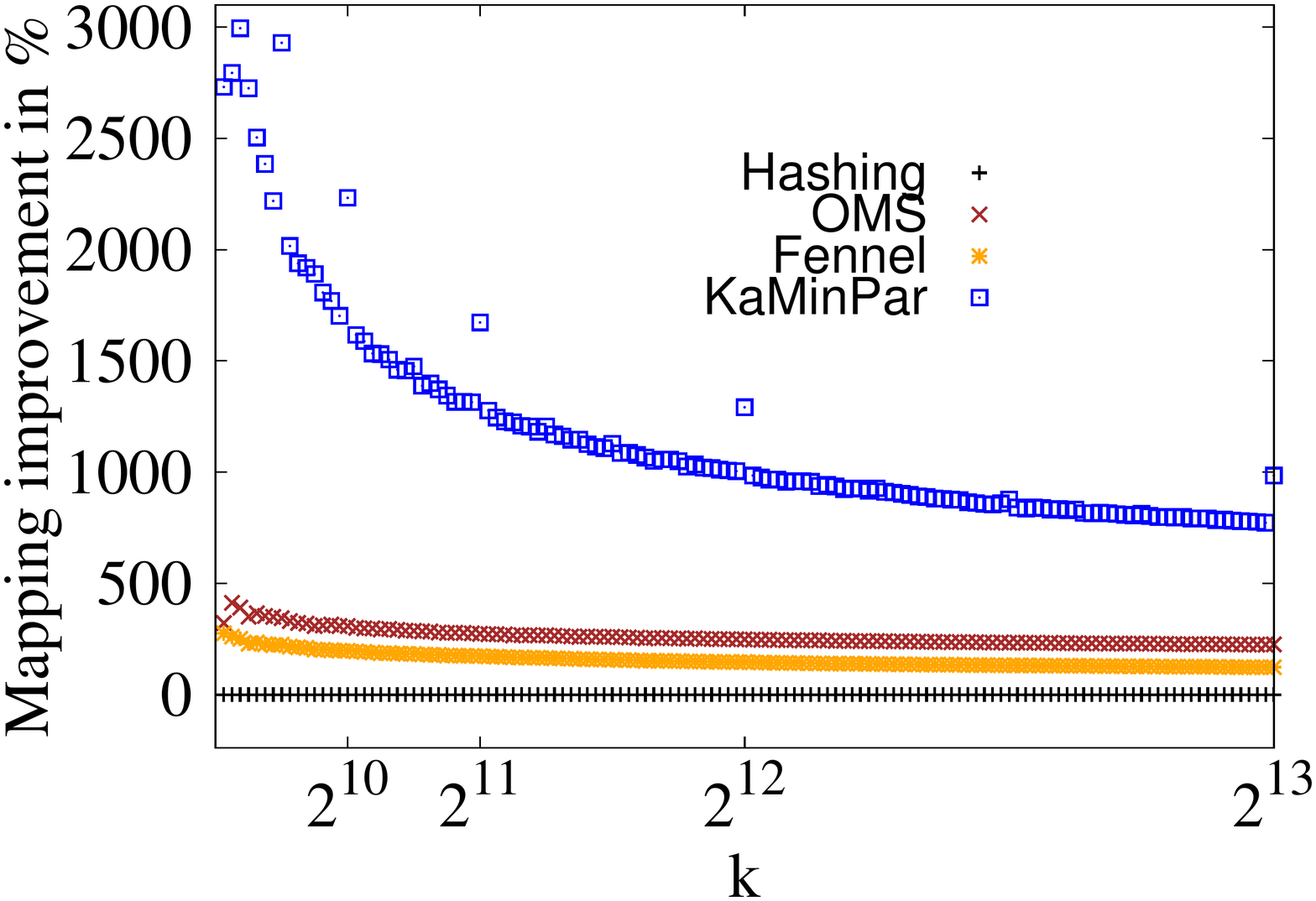}
		\vspace*{\capPositionSmall}
		\caption{Mapping improvement over Hashing.}
		\label{fig:stateoftheartMap_res}
	\end{subfigure}\hspace{2mm}%
	\begin{subfigure}[]{\scaleFactorSmall\textwidth}
		\centering
		\includegraphics[width=\imgScaleFactorSmall\textwidth]{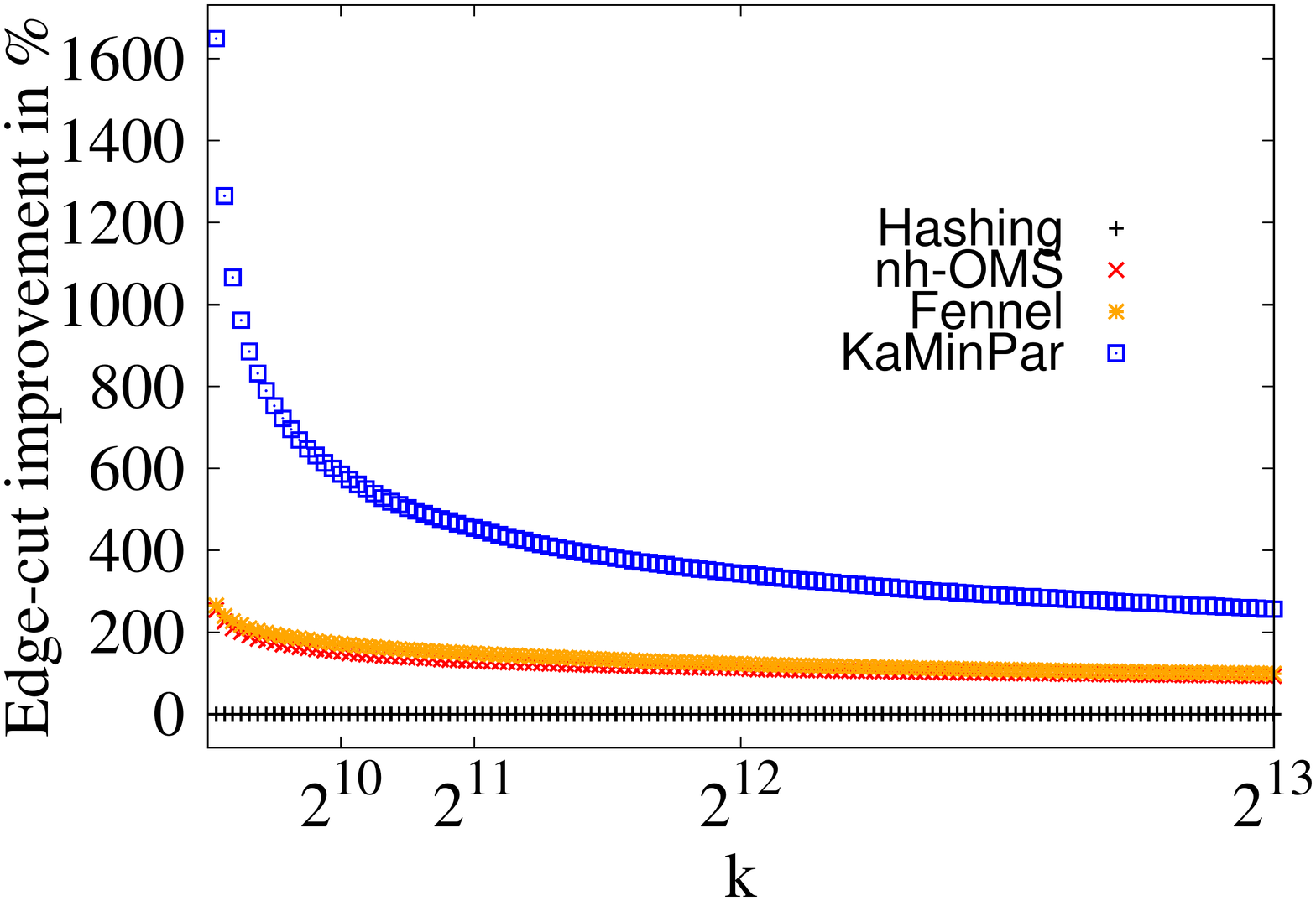}
		\vspace*{\capPositionSmall}
		\caption{Edge-cut improvement over Hashing.}
		\label{fig:stateoftheartPar_res}
	\end{subfigure}\hspace{2mm}%
	\begin{subfigure}[]{\scaleFactorSmall\textwidth}
		\centering
		\includegraphics[width=\imgScaleFactorSmall\textwidth]{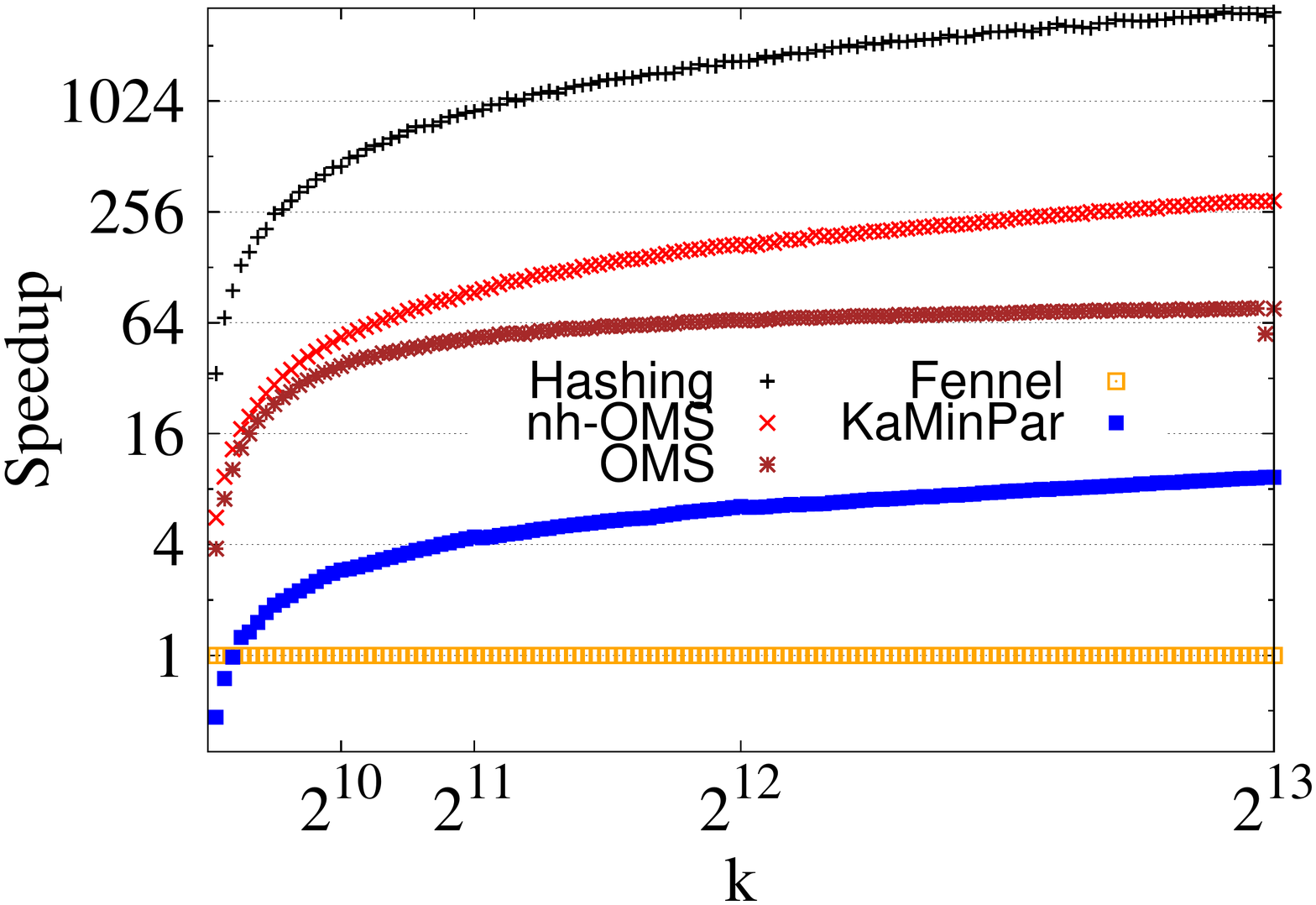}
		\vspace*{\capPositionSmall}
		\caption{Total speedup over Fennel.}
		\label{fig:stateoftheartPar_tim}
	\end{subfigure}
	\vspace*{\afterCapSmall}
	\begin{subfigure}[]{\scaleFactorSmall\textwidth}
		\centering
		\includegraphics[width=\imgScaleFactorSmall\textwidth]{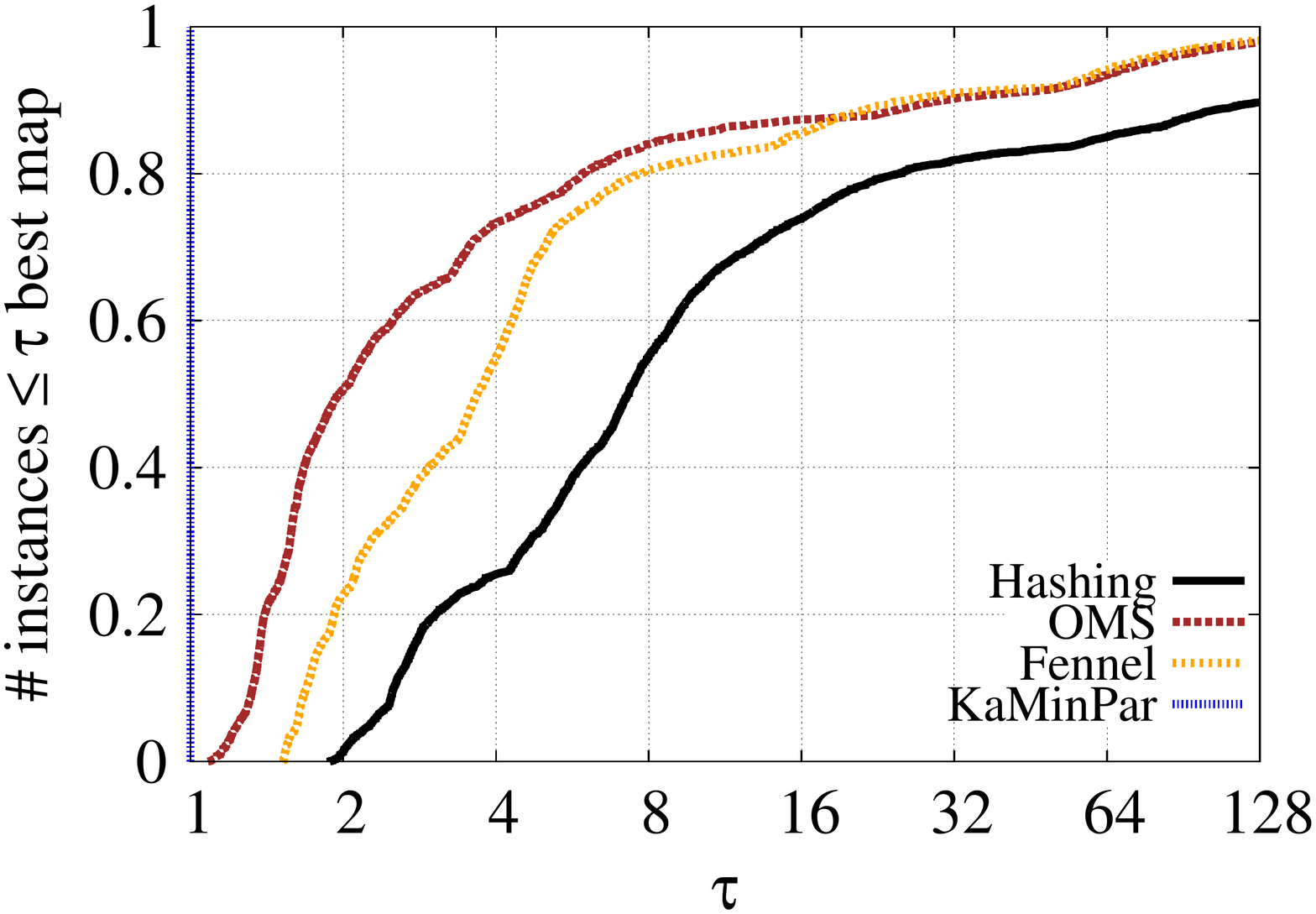}
		\vspace*{\capPositionSmall}
		\caption{Mapping performance profile.}
		\label{fig:stateoftheartMap_respp}
	\end{subfigure}\hspace{2mm}%
	\begin{subfigure}[]{\scaleFactorSmall\textwidth}
		\centering
		\includegraphics[width=\imgScaleFactorSmall\textwidth]{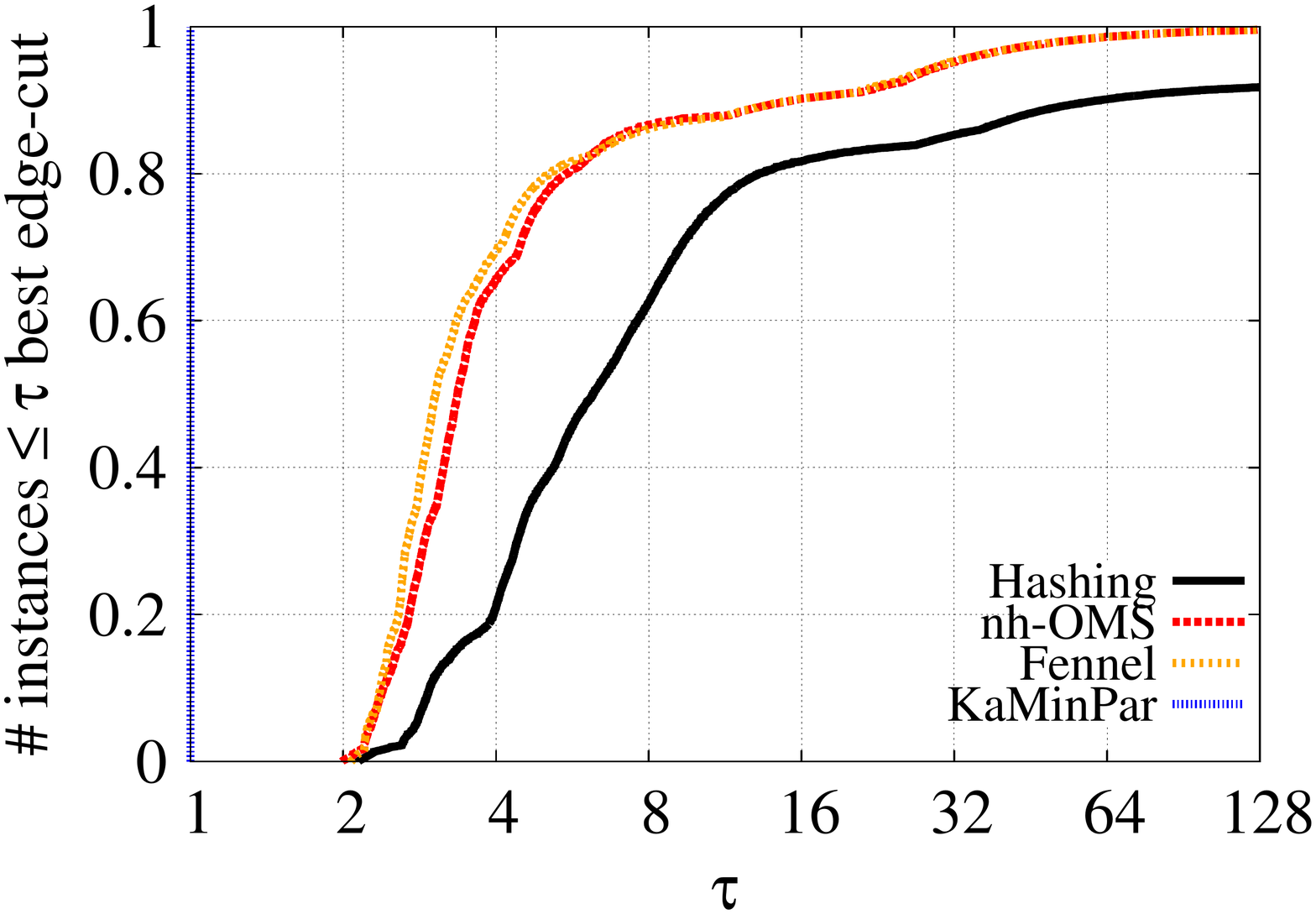}
		\vspace*{\capPositionSmall}
		\caption{Edge-cut performance profile.}
		\label{fig:stateoftheartPar_GPpp}
	\end{subfigure}\hspace{2mm}%
	\begin{subfigure}[]{\scaleFactorSmall\textwidth}
		\centering
		\includegraphics[width=\imgScaleFactorSmall\textwidth]{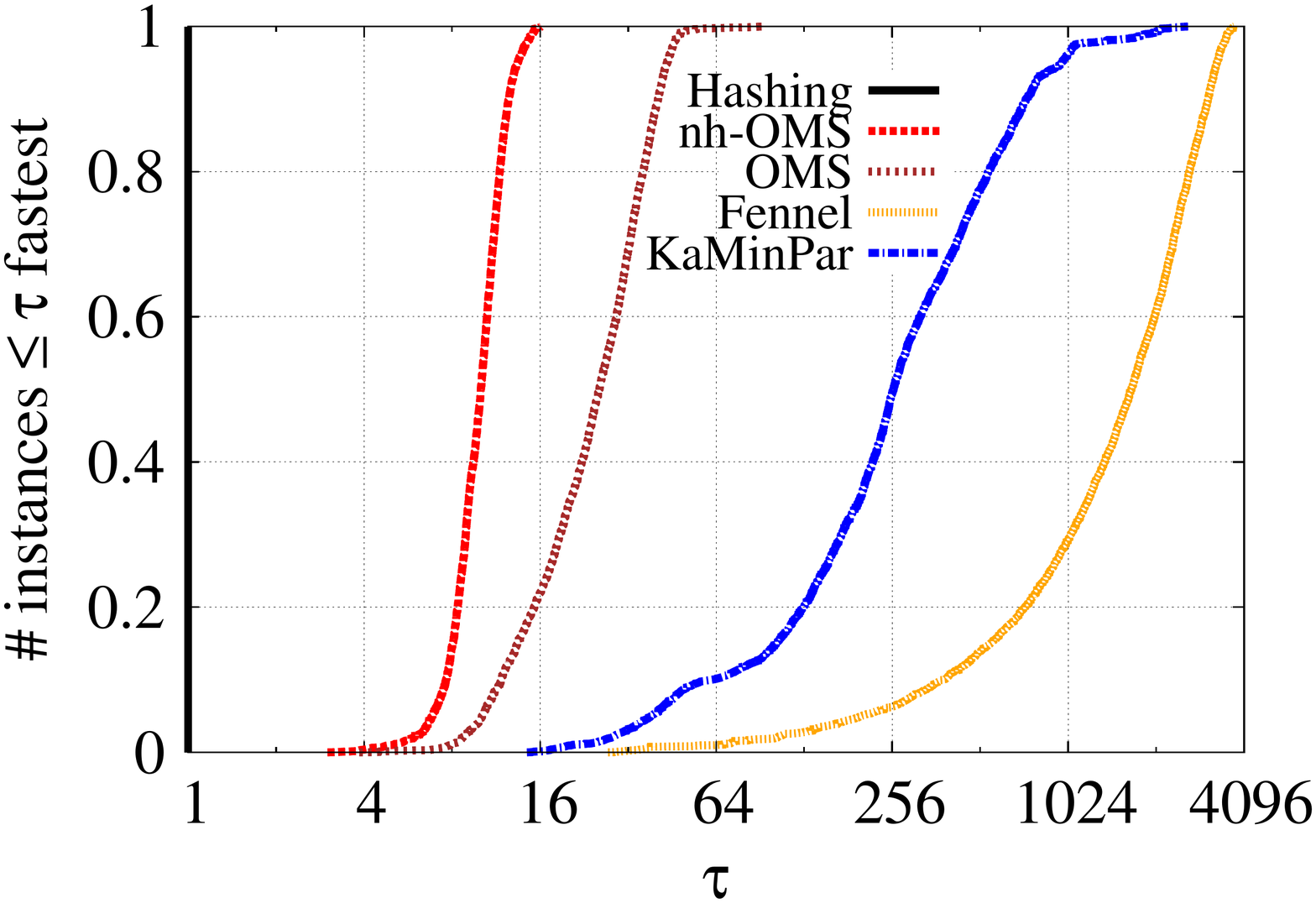}
		\vspace*{\capPositionSmall}
		\caption{Running time performance profile.}
		\label{fig:stateoftheartPar_timpp}
	\end{subfigure}

	\vspace*{.55cm}
	\caption{Comparison against the state-of-the-art. Higher is better.}
	\label{fig:state-of_the_art}
	
	\vspace*{-.5cm}
\end{figure*}

\vspace*{-.1cm}
\subsection{Scalability}
\label{subsec:Scalability}

\begin{table}[b!]
	\footnotesize
	\centering
	\setlength{\tabcolsep}{4pt}
	\begin{tabular}{|c|rr|rr|rr|rr|rr|}
		\hline
		\multirow{2}{*}{Threads} & \multicolumn{2}{c|}{Hashing} & \multicolumn{2}{c|}{nh-OMS} & \multicolumn{2}{c|}{OMS} & \multicolumn{2}{c|}{Fennel} & \multicolumn{2}{c|}{KaMinPar} \\
		
		& RT            & SU          & RT           & SU          & RT          & SU        & RT           & SU          & RT            & SU           \\
		\hline\hline
		1                        & 0.49          & 1.0         & 4.8          & 1.0         & 10.7        & 1.0       & 673.6        & 1.0         & 36.9          & 1.0          \\
		2                        & 0.72          & 0.7         & 4.6          & 1.1         & 6.4         & 1.7       & 346.3        & 1.9         & 19.3          & 1.9          \\
		4                        & 0.70          & 0.7         & 3.6          & 1.3         & 3.9         & 2.7       & 184.6        & 3.6         & 10.5          & 3.5          \\
		8                        & 0.72          & 0.7         & 3.0          & 1.6         & 2.6         & 4.1       & 96.3         & 7.0         & 5.8           & 6.4          \\
		16                       & 0.75          & 0.7         & 2.5          & 1.9         & 2.3         & 4.7       & 54.0           & 12.5        & 3.5           & 10.5         \\
		32                       & 0.46          & 1.1         & 1.7          & 2.8         & 1.3         & 8.2       & 44.2         & 15.2        & 3.1           & 11.9 \\\hline 
		
	\end{tabular}
        \vspace*{.25cm}
	\caption{Average running time in seconds~(RT) and average speedup~(SU) for $k=8192$.}
	\label{tab:scalability}
	
\end{table}
\begin{figure*}[t]
	\captionsetup[subfigure]{justification=centering}
	\centering
	\vspace*{-.75cm}
	\begin{subfigure}[t]{\scaleFactorSmall\textwidth}
		\centering
		\includegraphics[width=\imgScaleFactorSmall\textwidth]{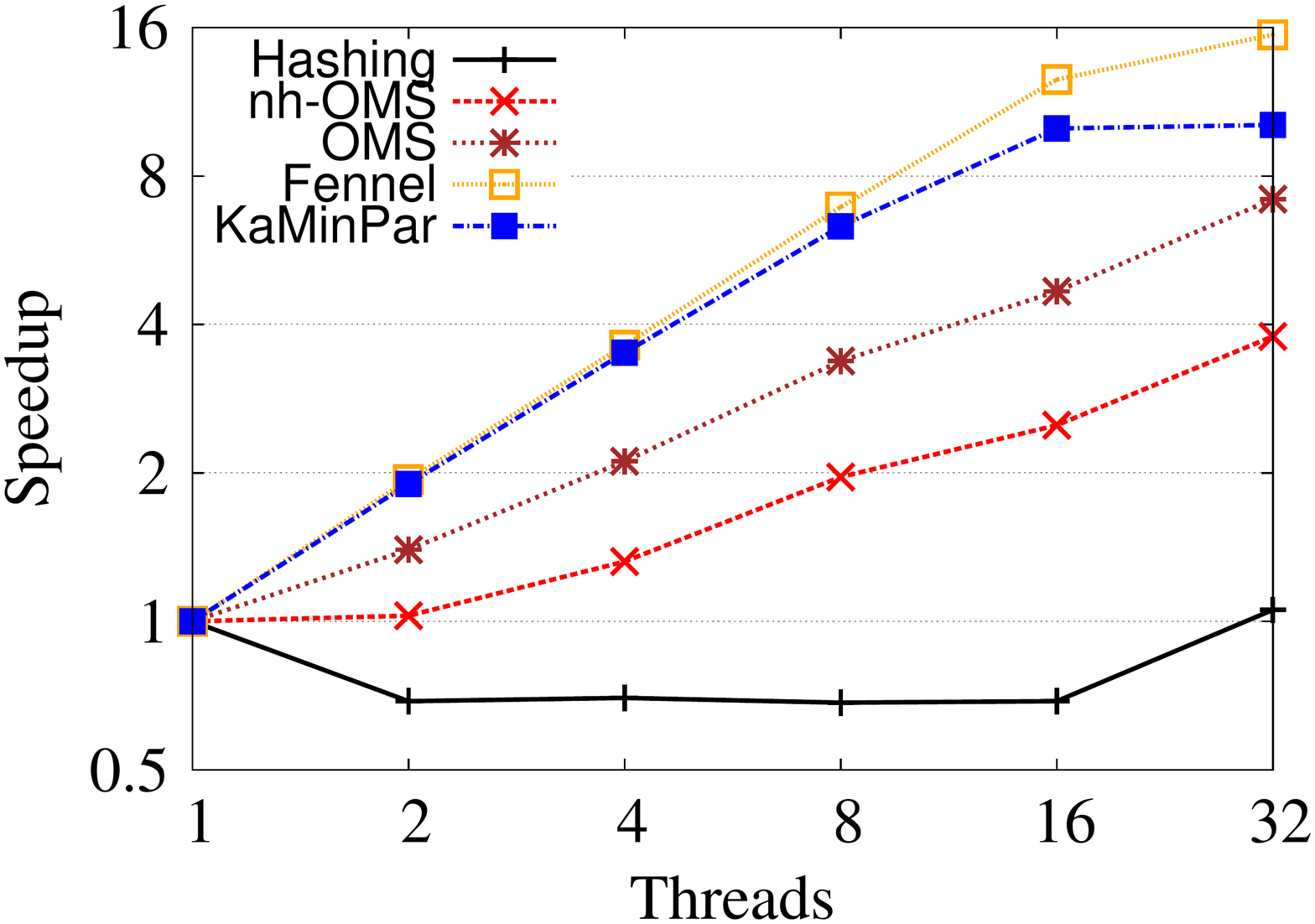}
		\vspace*{\capPositionSmall}
		\caption{Speedup versus number of used threads for graph soc-orkut-dir.}
		\label{fig:soc-orkut-dir_speedup}
	\end{subfigure}\hspace{2mm}%
	\begin{subfigure}[t]{\scaleFactorSmall\textwidth}
		\centering
		\includegraphics[width=\imgScaleFactorSmall\textwidth]{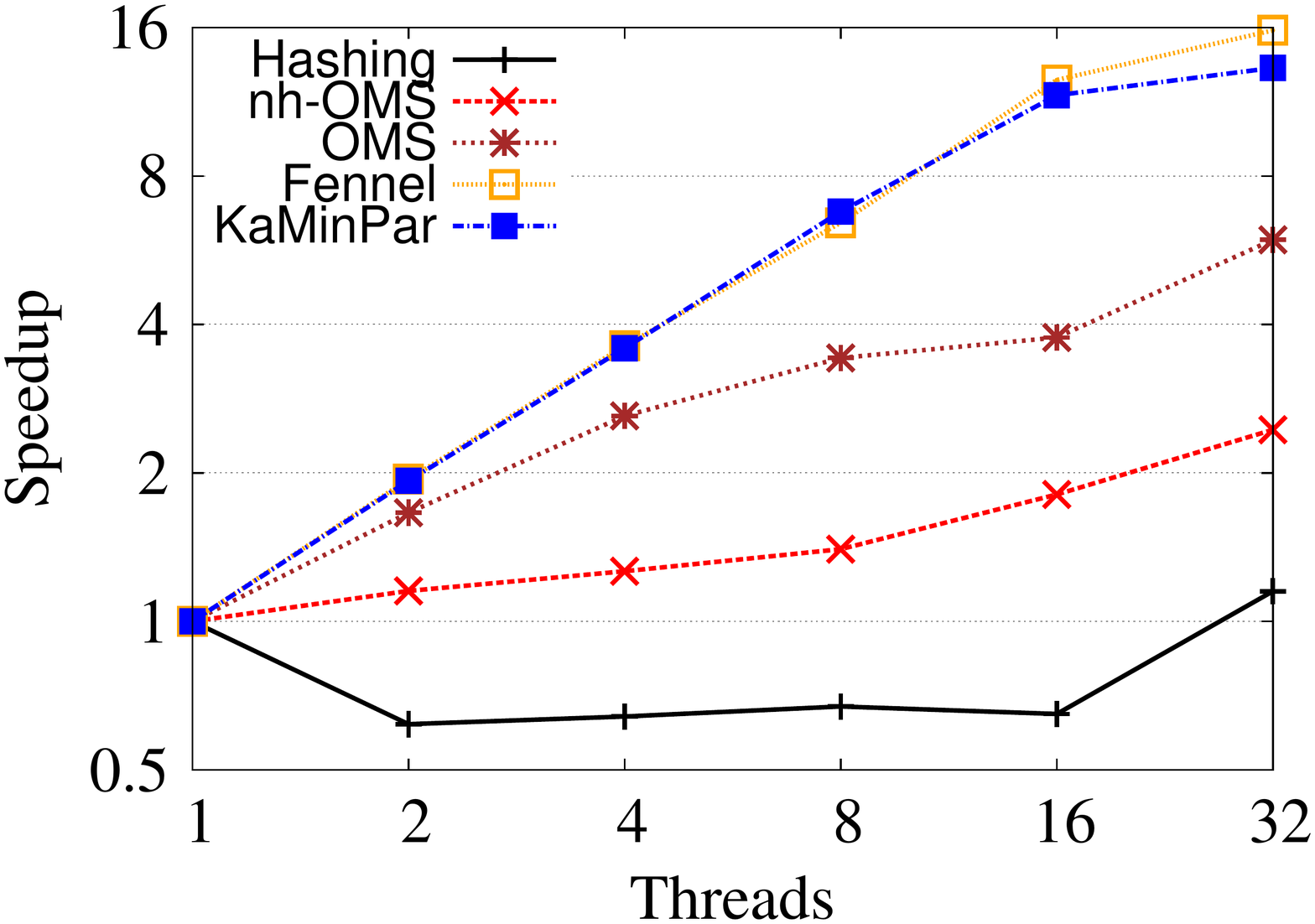}
		\vspace*{\capPositionSmall}
		\caption{Speedup versus number of used threads for graph HV15R.}
		\label{fig:HV15R_speedup}
	\end{subfigure}\hspace{2mm}%
	\begin{subfigure}[t]{\scaleFactorSmall\textwidth}
		\centering
		\includegraphics[width=\imgScaleFactorSmall\textwidth]{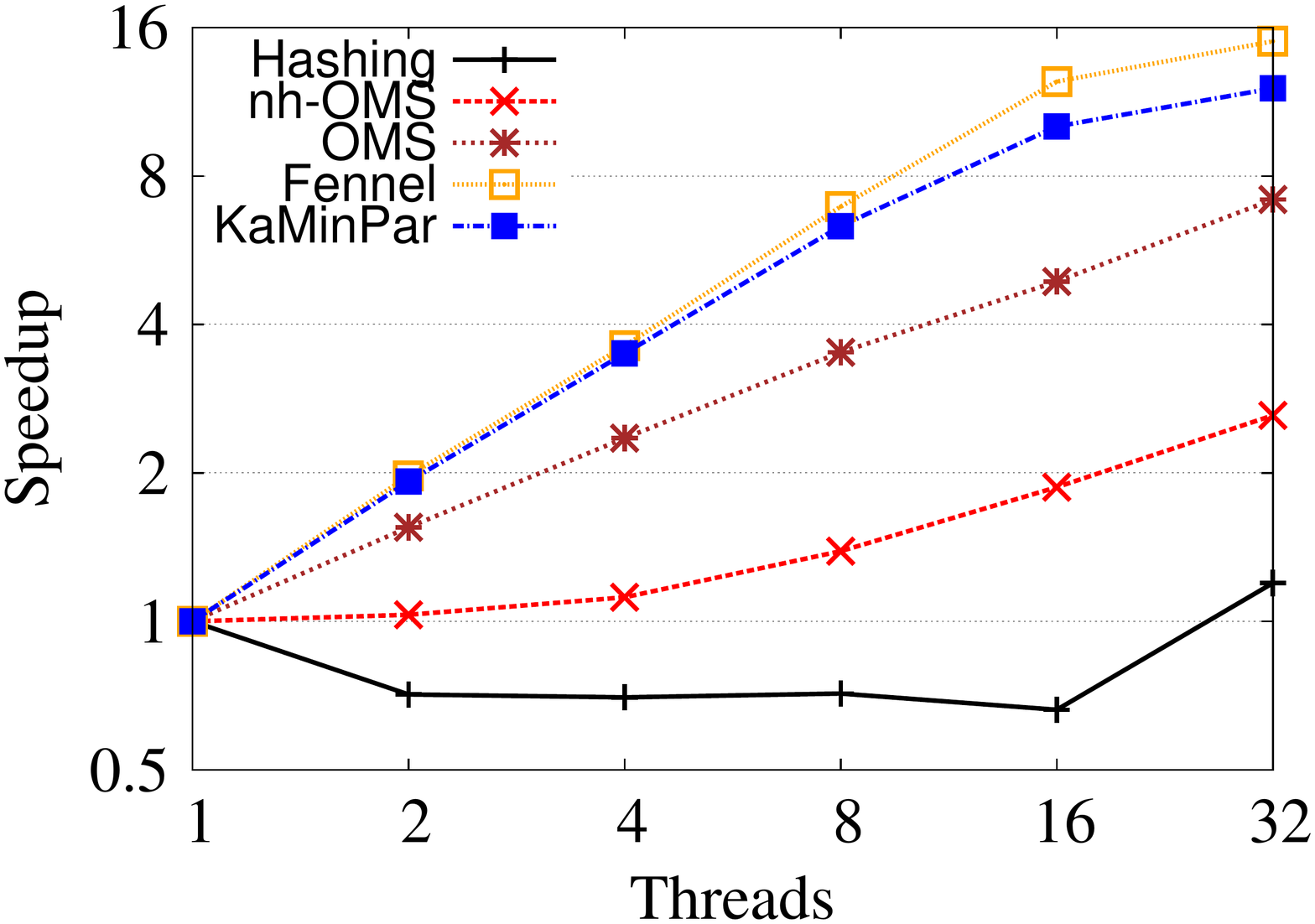}
		\vspace*{\capPositionSmall}
		\caption{Speedup versus number of used threads for graph soc-LiveJournal1.}
		\label{fig:soc-LiveJournal1_speedup}
	\end{subfigure}
	\vspace*{\afterCapSmall}
	\begin{subfigure}[t]{\scaleFactorSmall\textwidth}
		\centering
		\includegraphics[width=\imgScaleFactorSmall\textwidth]{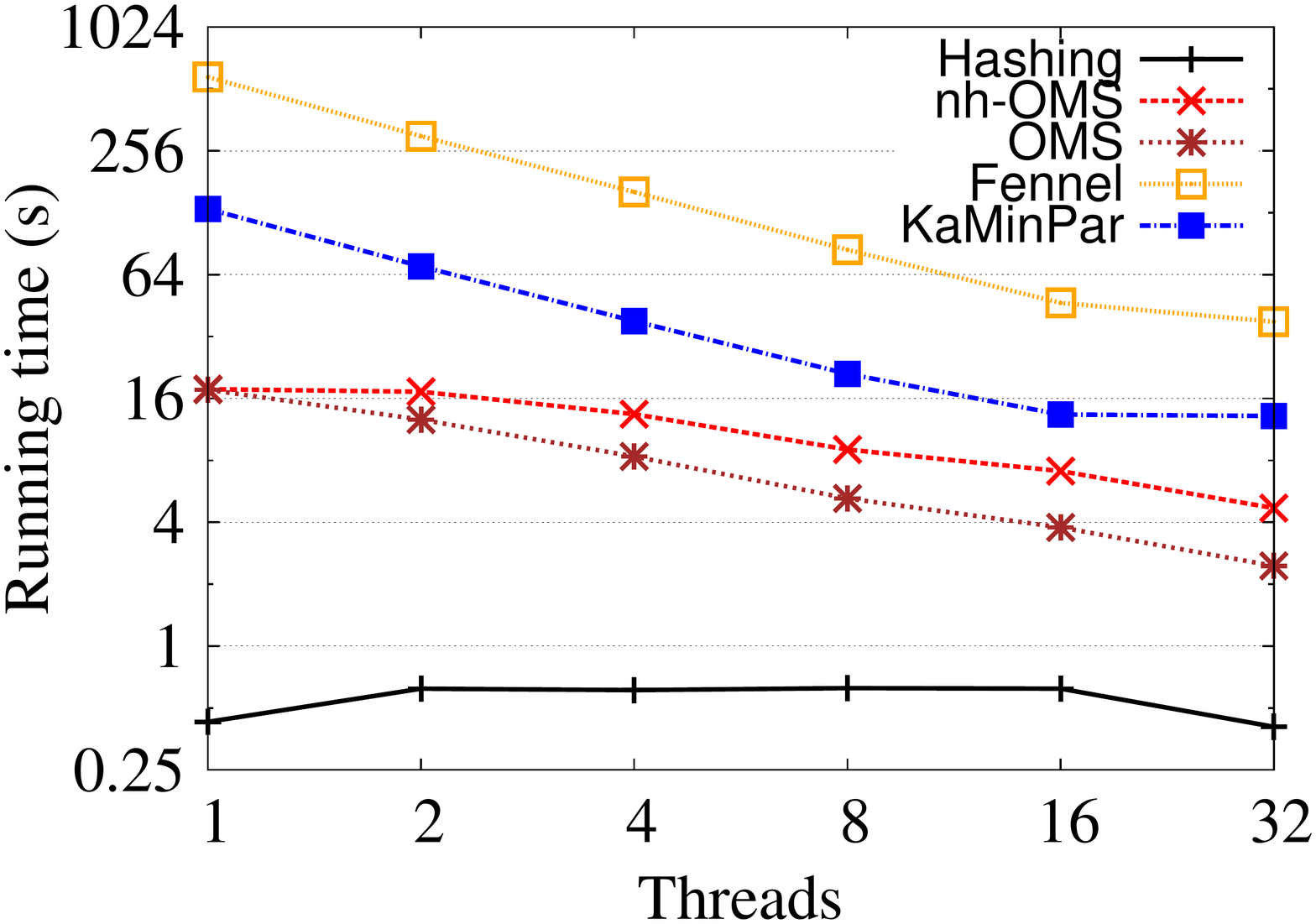}
		\vspace*{\capPositionSmall}
		\caption{Running time versus number of used threads for graph soc-orkut-dir.}
		\label{fig:soc-orkut-dir_times}
	\end{subfigure}\hspace{2mm}%
	\begin{subfigure}[t]{\scaleFactorSmall\textwidth}
		\centering
		\includegraphics[width=\imgScaleFactorSmall\textwidth]{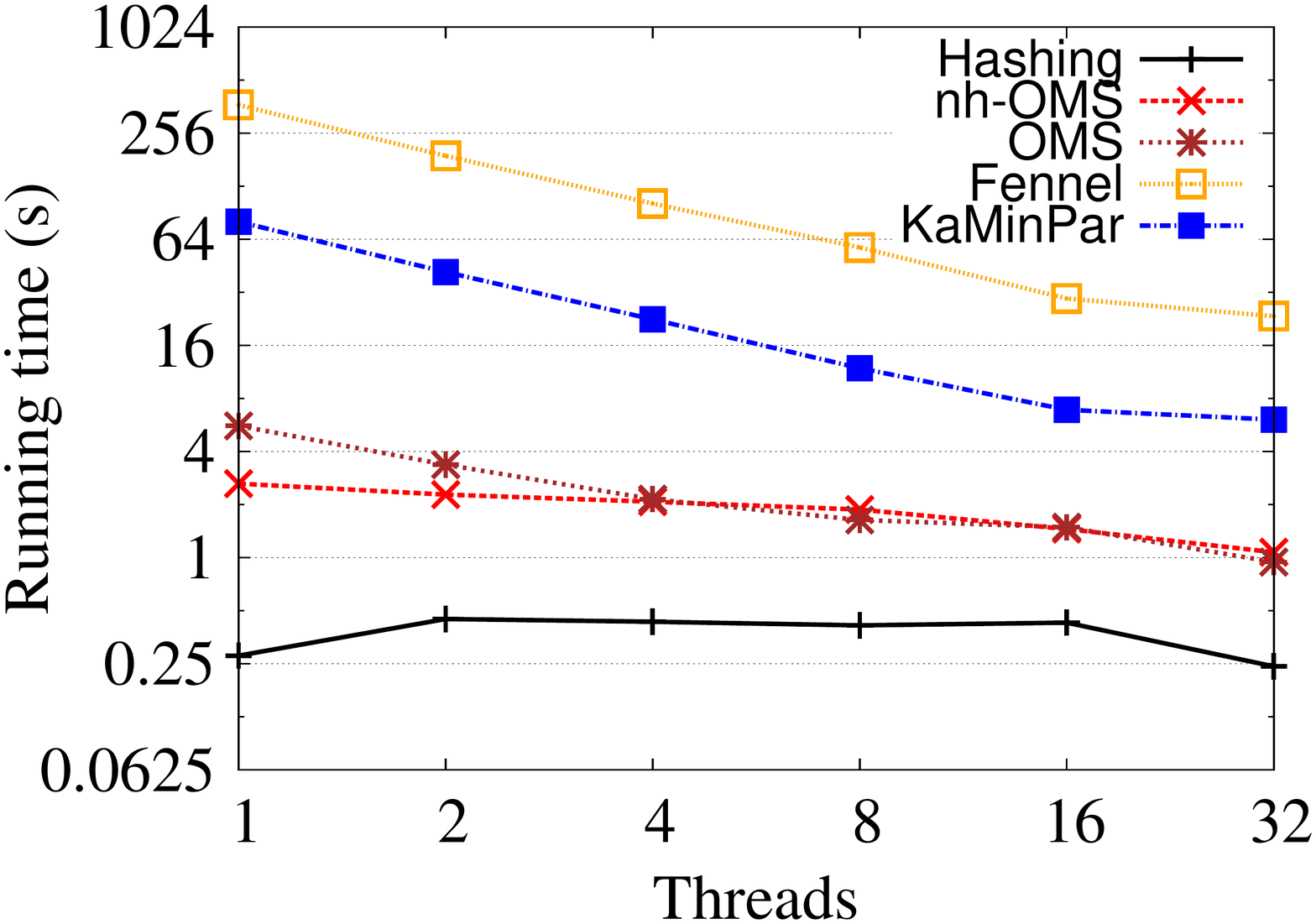}
		\vspace*{\capPositionSmall}
		\caption{Running time versus number of used threads for graph HV15R.}
		\label{fig:HV15R_times}
	\end{subfigure}\hspace{2mm}%
	\begin{subfigure}[t]{\scaleFactorSmall\textwidth}
		\centering
		\includegraphics[width=\imgScaleFactorSmall\textwidth]{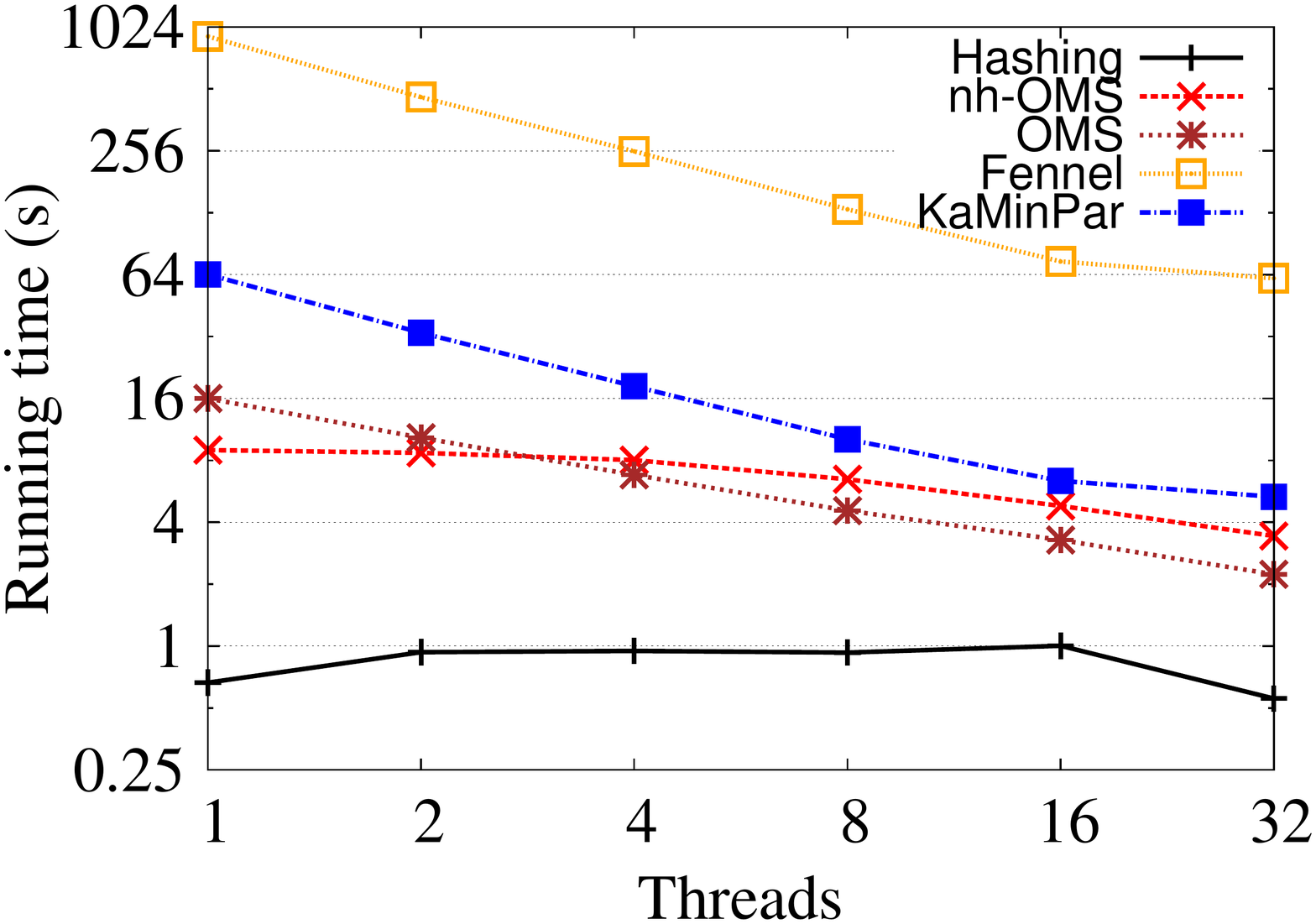}
		\vspace*{\capPositionSmall}
		\caption{Running time versus number of used threads for graph soc-LiveJournal1.}
		\label{fig:soc-LiveJournal1_times}
	\end{subfigure}

	\vspace*{.45cm}
	\caption{Speedup and time comparison for $k=8192$. Higher is better for speedup. Lower is better for time.}
	\label{fig:state-of_the_art_parallel}
	
	\vspace*{-.5cm}
\end{figure*}

Now we evaluate the scalability of our online recursive multi-section algorithm.
As in Section~\ref{subsec:State-of-the-Art}, we refer to our algorithm as~\emph{OMS} when a process mapping communication hierarchy is given and~\emph{nh-OMS} otherwise.
As competitors, we include KaMinPar, Fennel, and Hashing.
For a fair comparison against Fennel and Hashing, we implemented them with the same parallelization scheme of our algorithm, \ie a vertex-centric parallelization.
We do not include IntMap~\cite{fonseca2020better} in these experiments since it cannot run in parallel.
For these experiments, we selected the 12 graphs from the Test Set Table~\ref{tab:scalability} which have at least \numprint{2000000} nodes and partitioned them into $k=8192$ blocks using all algorithms.

In Figure~\ref{fig:state-of_the_art_parallel}, we plot speedup and running time versus number of threads for the graphs soc-orkut-dir, HV15R, and soc-LiveJournal1.
In Table~\ref{tab:scalability}, we plot the average running time in seconds and speedup over all graphs as a function of the number of threads.
For all graphs, Hashing presents the worst scalability, with speedups smaller than 1.
Although Hashing is theoretically an embarrassingly parallel algorithm, it has two limitations:
(i) it is extremely fast, hence the overhead of the parallelization is too large in comparison to the overall running time and
(ii) it neither reuses data nor accesses sequential positions in memory, so there are almost no cache hits.
On the other hand, Fennel presents the best scalability.
Differently than Hashing, it is rather slow but reuses data, \eg the assignments of previous nodes to blocks, and goes through the array of blocks in order to compute their score. 
Following Fennel, KaMinPar has the second best scalability which roughly reproduces the behavior reported in~\cite{gottesburen2021deep}.
The algorithms OMS and hr-OMS have an intermediary scalability between KaMinPar and Hashing.
This is explained by their characteristics, which are intermediary between Fennel and Hashing.
Note that OMS is more scalable than nh-OMS.
This happens because OMS goes through and scores several blocks in one of the partitioning subproblems contained in the multi-section hierarchy, which favors cache hits, whereas nh-OMS partitioning subproblems have at most $4$ blocks.
For 32 threads, the average running time of OMS is within a factor 3 of the running time of~Hashing.

\vspace*{-.5cm}
\section{Conclusion}
\vspace*{-.25cm}
\label{s:conclusion}

We proposed, analyzed, and engineered all the details of an online recursive multi-section algorithm to compute hierarchical partitionings of graphs. 
The complexity analysis shows that our algorithm is superior to current state-of-the-art one-pass streaming algorithms to partition graphs on the fly into a large number of blocks while also implicitly optimizing process mapping objectives.
To the best of our knowledge, this is the first streaming process mapping algorithm in literature.
We present extensive experimental results in which we tune our algorithm, explore its parameters, compare it against the previous state-of-the-art, and show that we can speed it up even further with a multi-threaded parallelization.
Experiments have shown that our algorithm is up to two orders of magnitude faster than the previous state of-the-art of streaming partitioning while producing solutions with better communication cost and slightly worse edge-cut.
On the other hand, our algorithm is only~$3$ times slower than Hashing when running on 32 threads while computing significantly better results.
While our algorithm is already useful for a wide-range of applications that need (hierarchical) partitions very fast, in future work, we plan to parallelize the algorithm in the distributed memory model and want to port it to GPUs.
Moreover, given the good results, we plan to publicly our algorithm release soon.

\bibliographystyle{splncs04}
\bibliography{phdthesiscs}

\end{document}